\documentclass[12pt]{article}
\usepackage{graphicx}
\usepackage{amssymb}
\usepackage{epstopdf}
\usepackage{color}
\usepackage{amsmath}
\usepackage{xcolor}
\usepackage{authblk}
\usepackage{pstricks}
\usepackage{tikz}
\usetikzlibrary{positioning}
\usepackage{amsthm}
\colorlet{linkequation}{blue}
\usepackage[colorlinks]{hyperref}

\newtheorem{proposition}{Proposition}[section]
\newtheorem{lemma}{Lemma}[section]
\newtheorem{thm}{Theorem}[section]

\newcommand{\bea}{\begin{eqnarray}}
\newcommand{\eea}{\end{eqnarray}}
\newcommand{\beq}{\begin{equation}}
\newcommand{\eeq}{\end{equation}}

\usepackage{hyperref}
\newcommand{\eps}{\epsilon}
\newcommand{\s}{\sigma}
\usepackage[top=2cm, bottom=2cm,left=1.5cm, right=1.5cm,headheight=1.4cm,headsep=1.5cm]{geometry}

\title{On an ordering-dependent generalization of Tutte polynomial }
\author{Joseph Ben Geloun\footnote{Max Planck Institute for Gravitational Physics, Albert Einstein Institute,
Am M\"{u}hlenberg 1, 14467 Golm, Germany} \footnote{Corresponding Author: jbengeloun@aei.mpg.de} ,
Francesco Caravelli\footnote{Invenia Labs, 27 Parkside Place, Parkside Cambridge CB1 1HQ and\\
London Institute of Mathematical Sciences, 35a South Street, London W1K 2XF, UK and\\
Department of Computer Science, University College London, Gower Street, London WC1E 6BT, UK }
}

\date{}
\begin{document}

\maketitle

\begin{abstract}
A generalization of Tutte polynomial involved in the evaluation of the moments of the integrated geometric Brownian in the Ito formalism is discussed. 
The new combinatorial invariant depends on the order in which 
the sequence of contraction-deletions have been performed on the graph.
 Thus, this work provides a motivation for studying an order-dependent Tutte polynomial in the context of stochastic differential equations. We show that in the limit of the control parameters encoding the ordering going to zero, the multivariate Tutte-Fortuin-Kasteleyn polynomial is recovered. 
\end{abstract}

\section{Introduction}
There is a deep connection between stochastic processes and quantum theory. As in quantum mechanics, an important role is played by the history of a process, in particular when the process is not ergodic. Moreover, if the noise of a certain stochastic process can be described by a Wiener process, then one can resort to many techniques developed for studying quantum field theories. This is true in particular for the case of the geometric Brownian motion.
The geometric Brownian motion (GBM) is the stochastic process described by the stochastic differential equation:
\begin{equation}
df=\mu f dt+\sigma f dW_t\,,
\label{eq:diff}
\end{equation}
where $W_t$ is a Wiener process and where $\mu, \sigma$ are constants, which parametrize the drift and the strength of the noise, respectively. The solution of equation (\ref{eq:diff}) can be written, formally, as
\begin{equation}
f(W_t,t)=\exp\left\{\left(\mu-\frac{\sigma^2}{2}\right)t+\sigma W_t\right\}.
\label{eq:moment}
\end{equation}
The Geometric Brownian motion is used for modelling many phenomena in a variety of contexts \cite{Gardiner,YorBook,Rev1,Rev2,Yor1}. 
This remains true for the case of the study of financial assets, where the distribution of returns can be approximated very often by a log-normal distribution \cite{YorBook,Gardiner}, at least in specific regimes. 
Yor provided various insights into the moments of the GBM, using as a starting point Girsanov theorem \cite{Girsanov}, 
and in particular in the case of the moments of its integral.

Recently, the properties of the moments of (\ref{eq:moment}) were studied by means of combinatorics in  \cite{Caravelli1}, obtaining a closed formula, although cumbersome, for generic moments of the functional. Moreover, these functionals were shown to be involved in evaluating the averages of the solution of a logistic, stochastic differential equations. These logistic differential equations appear in the problem of the optimal leverage when one considers the microstructure properties of financial markets \cite{Caravelli2}.

From the analysis provided in \cite{Caravelli1}, the central quantities of interest emerged to be the functions $s_k$, defined by $\langle f(W_t,t)^n\rangle=t^n \Gamma(n) s_n(\lambda_1,\cdots,\lambda_n)$, with $\lambda_i\equiv t(\mu+(k-i)\sigma^2)$. One of the main results of \cite{Caravelli1}, lays in the following proposition 
(Lemma 2 in \cite{Caravelli1}):

\

\noindent\textit{Proposition.} Under the condition of Ito's Wiener process, one has the following property for $s_k$:
\begin{equation}
s_k(\lambda_1,\cdots,\lambda_k)=\frac{e^{\lambda_k}}{\lambda_k}s_{k-1}(\lambda_1, \cdots, \lambda_{k-1})-\frac{1}{\lambda_k}s_{k-1}(\lambda_1, \cdots,\lambda_{k-2}, \lambda_{k-1}+\lambda_{k})\,. 
\label{eq:rec}
\end{equation}
with $s_0=1$. 

The recurrence relation \eqref{eq:rec} for the function $s_k$ can be solved to obtain a closed solution for the moments, and show that these functions can be written as a determinant of a linear operator \cite{Caravelli1}.
One interesting observation is that \eqref{eq:rec} is very suggestive of a deletion-contraction rule  performed on an abstract graph \cite{Tutte}. We are in the presence of a combinatorial object which can be described as a recursion relation of the form:
\begin{equation}
Z_{k}=\beta_{k} Z_{k-1}+\alpha_{k} Z_{k-1}' \,,
\end{equation}
where $Z_k$ and $Z_k'$ are functions, $k$ is related to the 
number of edges of some graph $G$, $\alpha_k$ and $\beta_k$ 
 coefficients associated with an edge of the graph being either contracted or deleted.  
In the present paper, we are interested in the mathematical properties of the relation \eqref{eq:rec} and explore in 
which sense one can actually consider this as the result
of an underlying generalized Tutte polynomial \cite{Tutte}.
This task presents several difficulties. As a first comment, we note that in order to describe the recurrence \eqref{eq:rec},
 one needs to formulate a combinatorial object similar to the multi-variate version of the 
Tutte polynomial \cite{Sokal}, and rather in the form of the 
Fortuin and Kasteleyn (FK) \cite{FK}. Introduced
in the study of random cluster models, the Tutte-FK polynomial on a graph $G(V,E)$, with edge weights $(p_e)_{e\in E}$, 
satisfies in fact a recurrence relation as
\bea\label{fk}
Z(G, (p_e)) = q_e\, Z(G-e, (p_{e' \ne e})) + p_e\, Z(G/e, (p_{e' \ne e}))  \,,
\eea   
where $p_e \in [0,1]$, $q_e$ is a ``dual'' weight chosen as $1-p_e$.
In full generality, $p_e$ and $q_e$ can be elements of
a commutative ring\footnote{The complete definition of
$Z(G)$ requires an extra parameter $\kappa$ that is not useful for our discussion.}.  We note that in \eqref{fk}, upon 
  contraction and deletion, $Z(G/e)$ and $Z(G-e)$ keep the form
of $Z(G)$ and become independent of the weights $p_e$ 
and $q_e$
of the edge $e$. This property is fundamental to 
ensure that the polynomial $Z(G)$ (as for any object falling in 
the universality class described by the Tutte polynomial) is
independent of the order in which one performs the sequence of contraction and deletion of the edges of $G$.  
From the right hand side of \eqref{eq:rec}, we note that the second function $s_{k-1}(\dots, \lambda_{k-1}+\lambda_k)$ still keeps its dependence of the variable 
$\lambda_k$. This suggests that the combinatorial
object that we seek is different from the usual Tutte-FK polynomial.  
If we want to extend the Tutte polynomial, it is therefore necessary to explicitly 
encode such a type of memory principle in the polynomial. 
As we will show, the memory encoded in the recursion relation in fact can be parametrized by the introduction of
an ordering of the edges $(e_{\sigma(1)}, e_{\sigma(2)}, \dots, e_{\sigma(n)})$, where $n=|E|$ and
$\sigma$ belongs to the symmetric group of $n$ elements.

Finally, we emphasize the fact that this work reports 
an extension of Tutte polynomial which depends on the order
of the contraction-deletion. Such a type of polynomial 
have not been the main focus of the attention
of combinatoricists, to the best of our knowledge. In \cite{BR}, Bollob\'{a}s and Riordan mention 
a possible extension of Tutte polynomial using noncommutative
ring elements which is contraction-deletion-ordering-dependent 
and perhaps worth to investigate in the future. In a different but still close perspective, 
this work precisely provides an instance in which this ordering for the Tutte
polynomial might be relevant in evaluating averages in exponentiated Wiener processes.   
We provide the generalized Tutte polynomial in Section \ref{sec:tutte}, followed by conclusions in Section \ref{sec:conc}. 
Appendix \ref{app:lemckl} provides the proof of a technical 
lemma and, to  illustrate the new graph invariant, we evaluate the
generalized Tutte polynomial on a specific graph in Appendix \ref{app:example}.

\section{The recurrence relation of the generalized Tutte}\label{sec:tutte}

The recurrence relation defined in \eqref{eq:rec} introduced in \cite{Caravelli1} can be thought as 
a modified contraction-deletion of the chain graph  (see Fig.\ref{graph1}). 
Our initial task is to generalize this recurrence relation to 
an arbitrary graph. Let us adopt the following considerations 
which clearly encompass the problem \eqref{eq:rec}: let $R$ be a commutative
ring and $T$ be a $R$-valued function
defined on a graph $G(V,E)$.  Let us assume that $T$   is defined on the set of edges weights 
$\{\lambda_e\}_{e \in E}$, $\lambda_{e}\in R$, and satisfies, for any edge $e$, 
\bea
T(G; \{\lambda_e\}) = \alpha(\lambda_e)\, T(G/e; \{\tilde\lambda_{\tilde{e} \ne e}\}) + \beta(\lambda_e)\,T(G-e; \{\lambda'_{e' \ne e}\}) \,, 
\eea
where $\tilde\lambda_{\tilde{e} \ne e} = \tilde\lambda_{\tilde{e} \ne e} (\{\lambda_e\})$ and $\lambda'_{e' \ne e}= \lambda'_{e' \ne e} (\{\lambda_e\})$ are new edge weights of $G/e$ and $G-e$, respectively, and are functions of the former edge weights of $G$; 
$\alpha$ and $\beta$ are $R$-valued functions.    

\
\begin{figure}[h]
\center
\begin{tikzpicture}
\draw[dotted]  (1.5,1) -- (3,1);
\draw (0,1) -- (1.5,1);
\draw (3,1)  -- (9,1) ; 
 \fill[fill=black] (0,1) circle (0.1);
 \fill[fill=black] (1.5,1) circle (0.1);
\fill[fill=black] (3,1) circle (0.1);
\fill[fill=black] (4.5,1) circle (0.1);
\fill[fill=black] (6,1) circle (0.1);
\fill[fill=black] (7.5,1) circle (0.1);
\fill[fill=black] (9,1) circle (0.1);


\draw (3,0)  -- (7.5,0) ; 
\draw[dotted]  (1.5,0) -- (3,0);
\draw (0,0) -- (1.5,0);

 \fill[fill=black] (0,0) circle (0.1);
 \fill[fill=black] (1.5,0) circle (0.1);
\fill[fill=black] (3,0) circle (0.1);
\fill[fill=black] (4.5,0) circle (0.1);
\fill[fill=black] (6,0) circle (0.1);
\fill[fill=black] (7.5,0) circle (0.1);
\fill[fill=black] (9,0) circle (0.1);

\draw (3,-1)  -- (7.5,-1) ; 
\draw[dotted]  (1.5,-1) -- (3,-1);
\draw (0,-1) -- (1.5,-1);
\fill[fill=black] (0,-1) circle (0.1);
 \fill[fill=black] (1.5,-1) circle (0.1);
\fill[fill=black] (3,-1) circle (0.1);
\fill[fill=black] (4.5,-1) circle (0.1);
\fill[fill=black] (6,-1) circle (0.1);
\fill[fill=black] (7.5,-1) circle (0.1);

\end{tikzpicture}
\put(-77,10){\footnotesize $\lambda_{n-1}$}
\put(-120,10){\footnotesize $\lambda_{n-2}$}
\put(-242,10){\footnotesize $\lambda_{1}$}
\put(-290,0){\footnotesize $G/e_n$}
\put(-86.5,40){\scriptsize $\lambda_{n-1}+\lambda_n$}
\put(-120,40){\footnotesize $\lambda_{n-2}$}
\put(-242,40){\footnotesize $\lambda_{1}$}
\put(-297,32){\footnotesize $G-e_n$}
\put(-30,70){\footnotesize $\lambda_n$}
\put(-77,70){\footnotesize $\lambda_{n-1}$}
\put(-242,70){\footnotesize $\lambda_{1}$}
\put(-286,60){\footnotesize $G$}
\caption{The chain graph $G$, $G-e_n$ and $G/e_n$,
and their weights.}
\label{graph1}
\end{figure}
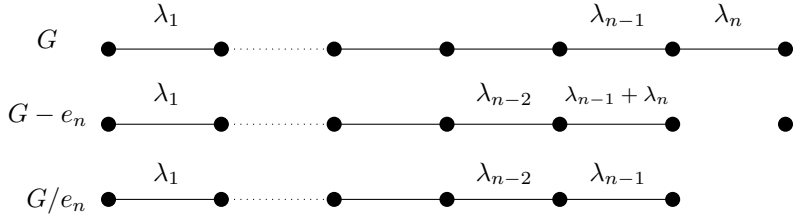

\

Performing successively the recurrence relation on a graph $G$, one ends with 
the boundary conditions (b.c.) $T(E_m; \emptyset) = q_m \in R$, 
where $E_m$ is 
the graph made with only $m$ vertices and no edges. 
We refer to these b.c. as terminal forms of the contraction-deletion. 
$T$ is therefore a sum of contributions of the rough form $\approx (\prod \alpha(\cdot))(\prod \beta(\cdot))T(E_m;\emptyset)$.  
At the end of the sequence
of contraction-deletion operations, each terminal form can be uniquely mapped to a spanning subgraph of $G$. The existence of that bijection is guaranteed by the fact
that the number of spanning subgraphs of $G$ is $2^{E}$,
and this is precisely the number  terminal forms\footnote{
The b.c. define the leaves in the abstract rooted tree of the contraction-deletion procedure with root $G$;  the two first nodes in this tree 
represent $G/e$ and $G-e$. Then we insert an edge between $G$  and $G/e$ and $G-e$ and so on. 
In this way, we define iteratively the set of contraction-deletions paths in $G$.}. 
If we start from an edge-labelled graph $G$, the bijection becomes
more apparent: starting from $G$ and following the unique path 
from the root to a given leaf $E_m$, we get a sequence of contracted and deleted edges of $G$. The unique spanning subgraph associated with the terminal form $E_m$ is the spanning subgraph of $G$ the edge set of which is the set 
of edges which have been contracted along the path to get $E_m$.  Given the discussion above, the following statement is therefore straightforward:
\begin{lemma}
The function $T(G;\{\lambda_e\})$ is a polynomial 
in the $\alpha$'s and $\beta$'s which admits a spanning subgraph 
expansion. 
\end{lemma}

The function $T(G;-)$  factorizes
along disconnected graphs, i.e. $T(G_1 \cup G_2; -) = T(G_1;-)T(G_2;-)$, if each contribution of spanning subgraph $A=A_1 \cup A_2$ of $G=G_1 \cup G_2$ will factorize between spanning subgraphs $A_1 \subset G_1$ and $A_2\subset G_2$.  
Another important operation on abstract graphs is called the vertex-union or one-point-join operation on graphs. Under this operation, the ordinary Tutte factorizes $T(G_1 \cdot G_2)=T(G_1)T(G_2)$.
In this work, we will check if these properties  are satisfied by the new graph polynomial we introduce.

We must be more specific about the type of functions
$\tilde\lambda_{\tilde e}$ and $\lambda'_{e'}$ that we will be interested in. 
Let us consider $R$  unital, a graph $G$ and its line graph $\mathcal L(G)$ \cite{Tutte}; consider moreover the
adjacency matrix of $\mathcal L(G)$ that, by a slight abuse of
notation, we write $A^{G}$. Given $G$ and
an edge $e$ of $G$, we introduce the functions:

\

$\tilde\lambda: E(G/e) \to R$, \; $\tilde\lambda(\tilde{e})=\lambda_{\tilde{e}} + \epsilon\, A^{G}_{\tilde{e}e}\lambda_e$\,; \quad $\epsilon \in R$\,;

\

$\lambda': E(G-e) \to R$, \; $\lambda'(e')=\lambda_{e'}
+ \epsilon' \,A^{G}_{e'e}\lambda_e$\,;\quad$\epsilon' \in R$\,.

\

The functions $\tilde\lambda$ 
and $\lambda'$ depend on a graph $G$ and one of its edges $e$;  these might introduce shifts on the weights of all remaining edges (of $G/e$ and of $G-e$) sharing vertices with $e$. As an example of the variable shift mentioned above, in Fig.\ref{fig:example} we show an example of deletion involving the line graph $A^G$: after each deletion (or contraction), the deleted variable modifies the edge variables of the resulting graph.

\begin{figure}
\centering
\includegraphics[scale=0.6]{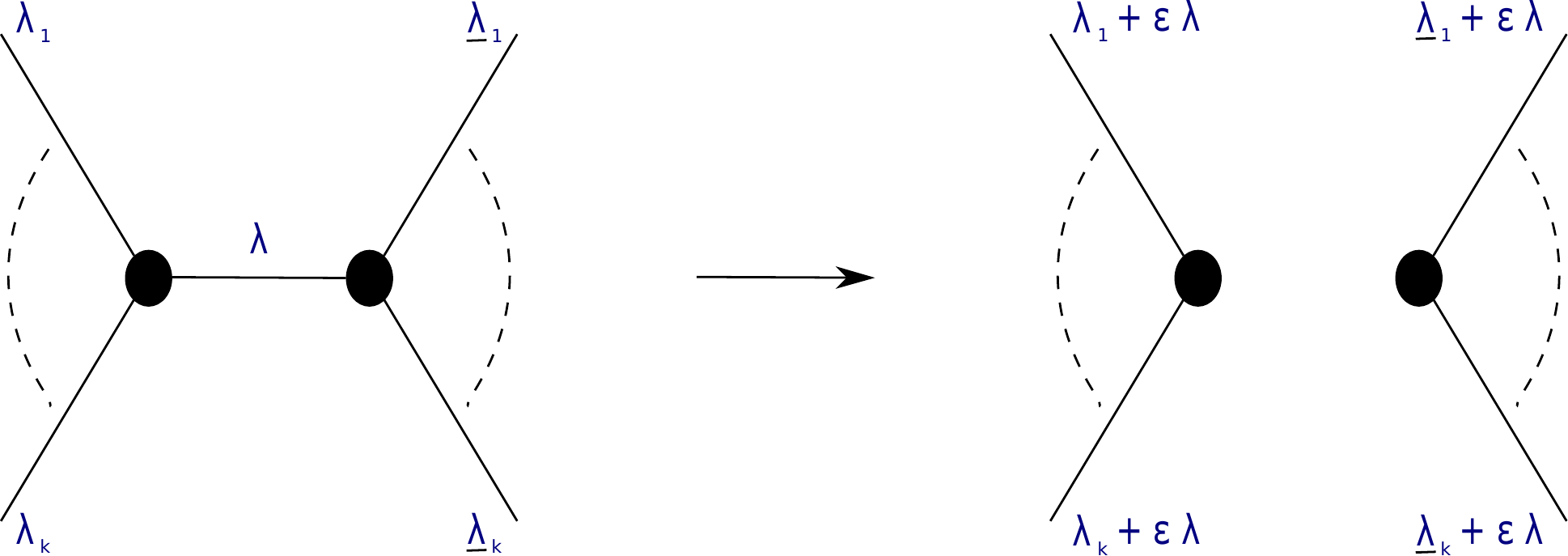}
\caption{An example of deletion
and ``attaching" the deleted edge variable to the nearby edge variables.
 A similar rule applies also to the contraction, with the replacement of $\epsilon^\prime$ instead of $\epsilon$.}
\label{fig:example}
\end{figure}

We are now in the position of introducing the problem treated in this work.
Let $G$ be a graph with edge set $E$,  $\epsilon, \epsilon' \in R$.
We are seeking for a function solution of the recurrence rule:
\bea\label{eq:REC}
\forall e \in E, \quad 
T^{\eps,\eps'}(G; \{\lambda_e\}_{e\in G})
&=&
 \alpha(\lambda_e) \,T^{\eps,\eps'}(G/e; \{\lambda_{\tilde{e}} + \epsilon\, A^{G}_{\tilde{e}e}\lambda_e\}_{\tilde{e}\in E(G/e)})
\cr
\cr
&+&  \beta(\lambda_e)\, T^{\eps,\eps'}(G-e; \{\lambda_{e'}
+ \epsilon' \,A^{G}_{e'e}\lambda_e\}_{e' \in E(G-e)})\,; \cr\cr
T^{\eps,\eps'}(E_m; \emptyset) &=& q_m \,. 
\eea
For $\epsilon=\epsilon^\prime=0$, this combinatorial object reduces to Tutte polynomial.
However, as we will see, the general case  requires extra care.

\subsection{Solving the ordering-dependent Tutte polynomial}

The general recurrence rule \eqref{eq:REC} is radically different from
an ordinary contraction-deletion rule because of the ordering-dependent
evaluation of the recurrence to specify $T^{\eps,\eps'}$.
Note that, as stated above, since the recurrence rule holds for any initial edge $e$,  it is not clear whether or not \eqref{eq:REC} has a general solution for arbitrary parameters  $\eps$ and $\eps'$.  However, 
if we specify an ordering in the edges, which is determined by a permutation $\sigma \in \mathfrak{S}(n)$, $(e_{\sigma(1)}, \dots, e_{\sigma(n)})$ and  if we perform a contraction and deletion of each edge in the order $e_{\s(1)} \to  e_{\s(2)} \to \dots \to e_{\s(n)}$, 
 then we claim that the problem \eqref{eq:REC}  has a solution if $T^{\eps,\eps'}$ depends on the ordering $\s$. 
To write an explicit expression of this solution, we first
 introduce compact notations. 
Note that, since the general case can be easily recovered from the following results, we will focus
on the simple case when $\s=id$, namely, when  the edges
are contracted and deleted according to the sequence $1 \to 2 \to 3 \to \dots \to n$.

We introduce, for a given edge $e_j$, 
\bea
G \cdot_j =  \left\{\begin{array}{l}
 G/e_j\,, \;\; \text{if $e_j$ is contracted } \\
G- e_j\,, \;\; \text{if  $e_j$ is deleted } 
\end{array}
\right.
\eea
so that, for successive operations, the notation  $G\cdot_1 \cdot_2  \cdots \cdot_k$, $k \leq n,$ provides a unique identification of the series of contractions or deletions. Then, given $(m,n) \notin \{1,\dots, k\}^{\times 2}$, 
we introduce 
\beq
A^{(k)}_{mn}=
A^{G\cdot_1 \cdot_2  \cdots \cdot_k}_{e_m,e_n} 
\eeq
the matrix element of the adjacency matrix of $\mathcal L(G\cdot_1 \cdot_2  \cdots \cdot_k)$ associated with vertices dual of $e_m$ and $e_n$. 
By convention, we set $A^{(0)} = A^G$.
In the same vein, simplifying the labelling of $\lambda_{e_j}$, we shall use $\lambda_j$. 

Consider an edge $e_j$ and a spanning subgraph $B$ of $G$, then define
the following function on $R$
\bea
\gamma_B(x_j) = \left\{\begin{array}{l}
\alpha(x_j)\,, \;\; \text{if $e_j \in B$ } \\
\beta(x_j)\,, \;\; \text{if  $e_j \notin B$}  
\end{array}
\right.
\eea
for $x_j=x(e_j)\in R$, and
\bea
\eps_B^{j} = \left\{\begin{array}{l}
\eps\,,\;\;  \text{if $e_j \in B$ } \\
\eps'\,, \;\; \text{if  $e_j \notin B$}
\end{array}
\right.
\eea

The following statement holds:
\begin{thm}\label{thm1}
Let $G$ be a graph with edge set $E$, $|E|=n$. 
Consider a labelling of edges of $G$ such that $\{e_j\}_{j=1,\dots, n}$,
and a family of edge weights $\{\lambda_{e_j}=\lambda_j\}_{j=1,\dots, n}$.
Let $q$, $\eps$ and $\eps'$ be elements in a commutative and unital ring $R$. 
A function $P^{\eps,\eps'}(G;q;\{\lambda_j\}) \in R$, solving the recurrence relation 
\bea
&& 
P^{\eps,\eps'}(G; q;\{\lambda_j\})
= \crcr
&& \alpha( \lambda_{1} ) \,P^{\eps,\eps'} ( G/e_{1};\,q;\,
\{\lambda_{ j } + \epsilon A^{(0)}_{ j 1 }\lambda_{1}\}_{j\ne 1}) 
+  \beta(\lambda_{1})\, P^{\eps,\eps'}(G-e_{1};\,q;\, \{\lambda_{j}
+ \epsilon' A^{(0)}_{j1}\lambda_{1}\}_{j \ne 1})\,; 
 \label{1rec}  \\\cr
&&
P^{\eps,\eps'}(E_m;\,q; \emptyset) = q^m \,, 
\label{bc}
\eea
is of the form
\bea\label{solu}
P^{\eps,\eps'}(G;q;\{\lambda_j\})
 = \sum_{B \subset G} q^{k(B)} \; \prod_{j=1}^n \gamma_B (\widehat\lambda_{j;B} \big(\{\lambda_l\}_{l\leq j}) \big) \,, 
\eea
where the sum is performed over the set of spanning subgraphs $B$ of $G$, $k(B)$ is the number of connected components of $B$, 
and, $\forall B$, one has
\bea\label{ccc}
\forall k\geq 1\,, &&
\widehat\lambda_{k;B} (\{\lambda_l\}_{l\leq k})) = \sum_{\ell =1}^{k} C_{k\ell;B}(G;\eps,\eps')\lambda_\ell\, ,   \cr\cr
\text{ for  } \, k=\ell\geq 1\, , &&
C_{kk;B}(G;\eps,\eps')= 1\, , \cr\cr
\text{ for  } \, k=\ell + 1\,, &&
C_{k(k-1);B}(G;\eps,\eps')= 
\mathcal{A}_{k (k-1);B} (Q_0)\,, \cr\cr
\forall k>\ell+1\geq 1\,, &&
C_{k\ell;B}(G;\eps,\eps')= \sum_{p=0}^{k-\ell-1}
\sum_{\stackrel{Q_p \subseteq \{\ell +1,\dots, k-1\}}{ |Q_p|=p }}
\mathcal{A}_{k\ell;B} (Q_p)\,, 
\eea 
where
\bea
&&
\text{for }\,  Q_0 = \emptyset, \;\; \mathcal{A}_{k\ell;B} (Q_0) = \eps_B^{\ell} A^{(\ell - 1)}_{k\ell} \,,
\cr\cr
&&
\text{for } \, Q_{p\geq 1} = \{j^Q_1,\dots, j^Q_p\}\neq \emptyset,\;
\; \ell + 1\leq j^Q_1< j_2^Q <\dots < j^Q_{p}  \leq k-1\,,  \cr\cr
&&
\mathcal{A}_{k\ell;B} (Q_p) = 
 \Big(\eps^\ell_B \prod_{a=1}^p \eps^{j^Q_a}_B \Big) 
\Big[A^{(j^Q_p-1)}_{k j^Q_p} \Big(\prod_{i=2}^p A^{(j^Q_{i-1}-1)}_{j^Q_i j^Q_{i-1}} \Big)A^{(\ell-1)}_{ j^Q_1 \ell} \Big]\;. 
\label{prodQ}
\eea
\end{thm}

\

To prove this proposition, we rely on the following
\begin{lemma}\label{lemCkl}
For all $k > 1$, and all $B$, $B\subset G$, 
\beq\label{bin}
C_{k1;B}(G;\eps,\eps') = \eps^{1}_B \sum_{\ell>1}^kA^{(0)}_{\ell 1}  C_{k\ell;B}(G;\eps,\eps') \,.
\eeq
\end{lemma}
\begin{proof}
 See Appendix \ref{app:lemckl}. 
\end{proof}
 
\proof[Proof of Theorem \ref{thm1}] 
We  note that the relation \eqref{bc} is a direct consequence of \eqref{solu}. 
We concentrate on the recurrence by applying the contraction and deletion to $e_1$. We proceed like in the proof of the regular Tutte polynomial case: consider all terms in \eqref{solu} associated with subgraphs of $G$
having $e_1$ and all terms associated with those which do not. We must prove that
\bea\label{innotin}
 \sum_{B \subset G;\, e_1\in B} q^{k(B)} \; \prod_{j=1}^n \gamma_B (\widehat\lambda_{j;B} (\{\lambda_l\}))  &=& \alpha( \lambda_{1} ) \,P^{\eps,\eps'} ( G/e_{1};q;\, 
\{\lambda_{ j } + \epsilon A^{(0)}_{ j 1 }\lambda_{1}\}_{j >  1})  \cr\cr
& =& \alpha(\lambda_1)\sum_{B  \subset G/e_1} q^{k(B)} \; \prod_{j=2}^{n} \gamma_B (\widehat\lambda_{j;B} (\{\lambda_l+ \epsilon A^{(0)}_{ l 1 }\lambda_{1}\})) \,,
\cr\cr
\sum_{B \subset G; \,e_1\notin B} q^{k(B)} \; \prod_{j=1}^n \gamma_B (\widehat\lambda_{j;B} (\{\lambda_l\}))  &=&   \beta(\lambda_{1})\, P^{\eps,\eps'}(G-e_{1};q;\, \{\lambda_{j}
+ \epsilon' A^{(0)}_{j1}\lambda_{1}\}_{j >  1})\cr\cr
& =&  \beta(\lambda_1)\sum_{B  \subset G - e_1} q^{k(B)} \; \prod_{j=2}^{n} \gamma_B (\widehat\lambda_{j;B} (\{\lambda_l+ \epsilon' A^{(0)}_{ l 1 }\lambda_{1}\})) \,.\cr\cr
&& 
\eea
Since the variable $q$ which describe the number of connected components present no particular difficulty, we focus on the products. 
We will prove the above statement in the following way: to each subgraph $B \subset G$, and for each factor in the product $\prod_{j=1}^{n} \gamma_B( \widehat\lambda_j (-))$, we assign a unique
and equal factor in the product $\gamma_B(\lambda_1)\prod_{j=2}^{n}  \gamma_{B'}( \widehat\lambda_j(-))$, labelled by a
subgraph $B' \in G\cdot_1$. 

Let us now consider the expansion,  for a given subgraph $B$, 
\bea
&&
\prod_{j=1}^n \gamma_B (\widehat\lambda_{j;B} (\{\lambda_l\})) 
= \gamma_B \Big(C_{1 1;B}(G;\eps,\eps')\lambda_1 \Big) 
\prod_{j=2}^n \gamma_B (\widehat\lambda_{j;B} (\{\lambda_l\}))  \cr\cr
&& = \gamma_B (\lambda_1) 
\prod_{j=2}^n \gamma_B \Big(C_{j 1;B}(G;\eps,\eps')\lambda_1 + 
\sum_{\ell > 1}^{k} C_{j\ell;B}(G;\eps,\eps')\lambda_\ell \Big) \,.
\label{factors}
\eea
Using Lemma \ref{lemCkl}, we expand $C_{j 1;B}$ in terms of $C_{j\ell}$
and write 
\bea
\prod_{j=1}^n \gamma_B (\widehat\lambda_{j;B} (\{\lambda_l\})) 
 &=& \gamma_B (\lambda_1) 
\prod_{j=2}^n \gamma_B \Big(   \eps^{1}_B \sum_{\ell>1}^j A^{(0)}_{\ell 1}  C_{j\ell;B}(G;\eps,\eps')\lambda_1 + 
\sum_{\ell > 1}^{k} C_{j\ell;B}(G;\eps,\eps')\lambda_\ell \Big) \cr\cr
& =& 
\gamma_B (\lambda_1) 
\prod_{j=2}^n \gamma_B \Big( \sum_{\ell>1}^j  C_{j\ell;B}(G;\eps,\eps')
\left( \lambda_\ell + \eps^{1}_B A^{(0)}_{\ell 1}  \lambda_1  \right)  \Big)
 \cr\cr
& =&
 \gamma_B (\lambda_1) 
\prod_{j=2}^n \gamma_B \Big(\widehat\lambda_{j} \big(\{\lambda_\ell + \eps^{1}_B A^{(0)}_{\ell 1}  \lambda_1  \}_{1 <\ell \leq j} \big)  \Big) \,.
\label{factors2}
\eea
Now, depending on whether $e_1$ belongs to $B$ or not, we can provide more informations on the initial
factor $\gamma_B(\lambda_1)$ in \eqref{innotin}. If $e_1 \notin B$, then 
$B \subset G$ is uniquely mapped to $B \in G-e_1$,  with the same
$\gamma_B (\widehat\lambda_{j>1})$; if $e_1 \in B$,
then $B\subset G$ is uniquely mapped to $B/e_1$, mapping 
 $\gamma_B (\widehat\lambda_{j>1}\big)$ to 
$\gamma_{B/e_1}(\widehat\lambda_{j>1})$. 
Thus \eqref{innotin} holds. 
\qed

 To obtain the $\s$-ordering-dependent polynomial $P^{\eps,\eps'}(G; \s;q;\{\lambda_j\})$,
the sums and products involved in the definition of $\widehat\lambda_{j}$ and $C_{k\ell }$ in \eqref{ccc} and \eqref{prodQ} must be handled a little bit differently. Indeed, as they stand in these equations, they depend explicitly on the ordering $\s = id$.   Reformulated in a full-fledged form using the ordering $\s(1)\to \dots \to \s(k)$, we can turn these symbols into $\widehat\lambda_{\s(j)}$ and $C_{\s(k)\s(\ell) }$. However, this turns out to be cumbersome in notations and does not add much to the discussion. So we will refrain to give their general formulae in their most expanded form.

To make clear how the generalized Tutte polynomial is evaluated, we have provided a worked out example  in App. \ref{app:example}.

\subsection{Reductions} We do have the following limiting cases:

-  the reduction to the Tutte-FK polynomial is direct by setting
\beq
Z(G;q; (\lambda_e)) =  
P^{\eps=0,\eps'=0}(G;q; \{\lambda_e\})\quad \text{with}\quad
\alpha(\lambda_e) = \lambda_e \,, \quad 
\beta(\lambda_e) = 1- \lambda_e\,, 
\eeq
with appropriate constraints $q_m = q^m$, $q \in R$, 
which represents an ordering-independent system. 

- The recurrence  \eqref{eq:rec} obeyed by function $s_k$ can be recovered 
by setting $G= C_n$ the chain graph with $n$ edges, 
\beq
s_n(\{\lambda_e\}) = 
T^{\eps=1,\eps'=0}(C_n; q=1;\{\lambda_e\})\quad \text{with}\quad
\alpha(\lambda_e) = \frac{e^{\lambda_e}}{\lambda_e}\,, \quad 
\beta(\lambda_e) =  -\frac{1}{\lambda_e}\,, 
\eeq
where we have made the choice that the deletion procedure
will be associated with the ordering-dependency of the
system. 

- Setting all weights to a constant, i.e. $\lambda_j = \lambda$, and
the ordering parameters $\eps$ and $\eps'$ to 1, one should get a 3-parameter deformation of the standard Tutte polynomial (the parameters will be $(q,\alpha(\lambda),\beta(\lambda))$).
The deformed invariant can be explicitly identified
from the above formulae and its properties will be investigated elsewhere. 

\subsection{Properties of the generalized Tutte}

We investigate now if the polynomial $P^{\eps,\eps'}$ factorizes along connected components. Consider the graph $G=G_1 \cup G_2$ made of disconnected 
graphs $G_1$ and $G_2$, with sets $E_1$ and $E_2$ 
of edges, respectively. Let us denote $|E_1|=n_1$ and $|E_2|=n_2$.  Given an ordering $\s$ of the edges 
of $G$, we want to infer two orderings $\s_1\in\mathfrak{S}_{n_1}$ and $\s_2\in\mathfrak{S}_{n_1}$ of the edges in the two separate graphs $G_1$ and $G_2$, respectively. A simple listing achieves this: construct
the image of $\s_1$ by listing all edges of $G_1$ in their order of appearance in $\s$, and then choose the image of $\s_2$ to be the rest of the edges (of $G_2$) keeping again their order of appearance in $\s$. The two permutations $\s_1$ and 
$\s_2$ will be called induced orderings (on $E_1$ and $E_2$) from $\s$. 

Then, we have  
\bea
&&
P^{\eps,\eps'}(G;\s=id;q;\{\lambda_j\})
 = \sum_{B_1\cup B_2 \subset G_1 \cup G_2}
q^{k(B_1)}q^{k(B_2)}
\prod_{j=1}^{n_1+n_2}\gamma_{B_1 \cup B_2}(\widehat\lambda_j(\{\lambda_{l}\} )) \;; 
\cr\cr
&& 
\prod_{j=1}^{n_1+n_2}\gamma_{B_1 \cup B_2}(\widehat\lambda_j(\{\lambda_{\s(l)}\}) )
 = \Big[\prod_{j=1}^{n_1}\gamma_{B_1 \cup B_2}(\widehat\lambda_{\s_1(j)}(\{\lambda_{l}\}) )\Big]\Big[\prod_{j=1}^{n_2}\gamma_{B_1 \cup B_2}(\widehat\lambda_{\s_2(j)}(\{\lambda_{l}\}) )\Big].
\eea
For a disconnected graph $G_1 \cup G_2$,
there is an ordering of the edges of $G_1 \cup G_2$ such that
the matrix $A^{(k)}$ of $G_1 \cup G_2$ will always appear block diagonal
(an edge of $G_1$ cannot have a common vertex with 
any edge of $G_2$ after an  arbitrary number of contractions and deletions). This implies that the contraction or 
deletion of an edge in a graph, say $G_1$, does not affect the 
variables in the other graph $G_2$ (and vice versa). 
For any $k\leq n_1 + n_2$, 
there exists $k_1$ and $k_2$, such that $k=k_1 + k_2$ 
and 
\beq
(G_1 \cup G_2 )\cdot_1 \cdot_2 \dots \cdot_{k}
 = (G_1\cdot_{\s_1(1)} \cdot_{\s_1(2)} \dots \cdot_{\s_1(k_1)})
\cup (G_2\cdot_{\s_2(1)} \cdot_{\s_2(2)} \dots \cdot_{\s_2(k_2)}) \,. 
\eeq

We stress now two important points:  
(1) the definitions of $\widehat{\lambda}_{\s_1(j)}$
and $\widehat{\lambda}_{\s_2(j)}$ can be now 
restricted to the adjacency matrix of the line graphs of
$G_1$ and $G_2$, respectively, and (2), given $i\in \{1,2\}$, the function
$\gamma_{B_1 \cup B_2} (\widehat{\lambda}_{\s_{i}(j_i)})$,
$j_i=1,\dots, n_i$,  becomes independent of $B_{i'\ne i}$
and it is exactly the same as $\gamma_{B_i} (\widehat{\lambda}_{\s_{i}(j
_i)})$, $j_i=1,\dots, n_i$, $i=1,2$. Therefore, the following proposition
holds: 
\begin{proposition}[Factorization under disjoint union]
Let $G_1$ and $G_2$ two graphs. Consider an ordering of the
edges of $G_1 \cup G_2$, then there exist two induced orderings 
of the edges, $\s_1$ of $G_1$ and $\s_2$ of $G_2$, 
and  
\beq
P^{\eps,\eps'}(G_1 \cup G_2;\s=id;q;\{\lambda_j\}_{j=1}^{n_1+n_2})
 =P^{\eps,\eps'}(G_1;\s_1;q;\{\lambda_j\}_{j=1}^{n_1})
P^{\eps,\eps'}(G_2;\s_2;q;\{\lambda_j\}_{j=1}^{n_2})\,.
\eeq
\end{proposition}
Like the Tutte polynomial, 
the polynomial $P^{\eps,\eps'}$ obeys the
factorization property under disjoint union operation.
It is also immediate that given two orderings in $G_1$ 
and $G_2$, one can trivially construct an ordering of edges of $G_1 \cup G_2$ (by concatenating these)
such that the above equality still holds. 
However, the factorization property 
under one-point-join operation will  fail for 
$P^{\eps,\eps'}$ because joining two graphs 
by a vertex affects the structure of the adjacency
matrices. Simple  counter-examples exist can be easily checked by the 
reader.

\subsection{Towards a noncommutative Tutte polynomial}

Bollob\'as and Riordan \cite{BR}
mention the likely existence of a Tutte polynomial on 
a noncommutative ring. We comment here another possibility to 
identify a noncommutative Tutte polynomial 
using a noncommutative/nonassociative algebra of
weight functions which is based on the 
generalized polynomial found in this work. 
For simplicity, and because we need an  integration measure,
 we will restrict to $R=\mathbb{R}$, the field of real numbers. 
The discussion below is purely formal. However, it can certainly be made 
rigorous using the appropriate functional space.

The polynomial $P^{\eps,\eps'}(G;-)$ involves several convolutions of functions $\alpha$ and $\beta$ that we symbolize by $\gamma_B$. Indeed, we can introduce a  family of kernels $\{K_B\}_{B\subset G}$ such that 
\bea
P^{\eps,\eps'}(G;q;\{\lambda_j\}) &=& \sum_{B\subset G} q^{k(B)}
\int_{\tilde\lambda_j\in \mathbb{R}}[\prod_jd\tilde\lambda_j]\; \; K_B(\{\lambda_j\}; \{\tilde\lambda_j\})\gamma_B(\tilde\lambda_1)\cdots \gamma_B(\tilde\lambda_n) \cr\cr
& =& 
\sum_{B\subset G} q^{k(B)} \; 
(\gamma_B^1 , \gamma_B^2,\dots, \gamma_B^n )_{\star_{K_B}} (\{\lambda_j\}_{j\geq 1})\,, 
\eea
where $\star_{K_B}$ is a $n$-ary law with kernel $K_B$ which convolutes $n$ functions $\gamma_B^{j} \in \{\alpha,\beta\}$, and which gives as an output a function. The kernel $K$ can be thought of
as a tensor product of delta-distributions which enforces the value 
of the $j$th component $\gamma(\tilde\lambda_j)$ to the corresponding value appearing in the
general formula \eqref{solu}. 
In general, $n$-ary laws are nonassociative (bracketing matters) and noncommutative
(ordering the functions matters). In the above specific instance, because the kernel $K_B$ can be factorized 
\beq
K_B= \delta(\tilde\lambda_1 - \lambda_{1}) \delta(\tilde\lambda_2 - \widehat\lambda_{2;B}(\lambda_1,\lambda_2) )\dots \delta(\tilde\lambda_n - \widehat\lambda_{n;B}(\{\lambda_{l}\}_{l\leq n}) )\,,
\eeq
the above $\star_{K_B}$-product factorizes in binary laws, 
\bea
(((\gamma_B^1 \star_{K_B;1} \gamma_B^2) \star_{K_B;2} \gamma_B^3) \dots ) \star_{K_B;n-1} \gamma_B^n \,,
\eea
where the $\star_{K_B;j} $-product is associated with the kernel 
\bea
\delta\big(\tilde\lambda_1 - \lambda_{1}) \delta(\tilde\lambda_2 - \widehat\lambda_{2;B}(\lambda_1,\lambda_2) )\dots \delta(\tilde\lambda_{j+1} - \widehat\lambda_{j+1;B}(\{\lambda_{l}\}_{l\leq j+1})\big)\,. 
\eea
 In this sense, an ordering-dependent Tutte polynomial
would be realized in terms of a noncommutative algebra of functions 
endowed with several binary laws and it will obey 
\bea
&& 
P^{\eps,\eps'}(G; q;\{\lambda_j\})
=
 \Big[\alpha \star_1 \,P^{\eps,\eps'} ( G/e_{1};\,q;\,
\{\lambda_{ j } \}_{j\ne 1}) \Big] (\lambda_1)
+  \Big[\beta\star'_{1}\, P^{\eps,\eps'}(G-e_{1};\,q;\, \{\lambda_{j}\}_{j \ne 1})
\Big] (\lambda_1)\cr\cr
&& 
 = 
 \alpha( \lambda_{1} ) \,P^{\eps,\eps'} ( G/e_{1};\,q;\,
\{\lambda_{ j } + \epsilon A^{(0)}_{ j 1 }\lambda_{1}\}_{j\ne 1}) 
+  \beta(\lambda_{1})\, P^{\eps,\eps'}(G-e_{1};\,q;\, \{\lambda_{j}
+ \epsilon' A^{(0)}_{j1}\lambda_{1}\}_{j \ne 1})\,.
\eea
up to b.c.. Note that we have made a choice of the kernels of
$\star_1$ and $\star'_1$  in order to write the convolution $\alpha \star_1 \,P^{\eps,\eps'}(G/e_1;-)$ and $\beta \star'_1 \,P^{\eps,\eps'}(G-e_1;-)$ 
in the appropriate form to match the result. We could have chosen 
a different kernel and impose  that it is $P^{\eps,\eps'}(G/e_1;-)\star_1 \alpha$ and $P^{\eps,\eps'}(G-e_1;-)\star_1 \beta$ which give
the correct answer. Thus we can choose to 
convolute the functions $\alpha$ and $\beta$ either on
the left or on the right of the reduced polynomials and this exhibits the richness of this framework.

\subsection{A realization with graph with half-edges}

Graph with half-edges are useful in the context of Quantum Field Theory. 
Half-edges represent external fields or probes used to measure
higher energetic processes. 
Recently this type of graphs have investigated in several contexts
(QFT, matrix models ribbons graphs, tensor graphs) \cite{Krajewski:2008fa,Krajewski:2010pt,avo}. 
In particular, in \cite{avo}, Avohou et al. discussed an extension of Tutte polynomial 
to this type of graphs which satisfies a recurrence relation called
``contraction and cut'' relation. The cut operation on an edge $e \in G$ (denoted $G \vee e$) is intuitive: one deletes the edge and let  two half-edges at its end vertex or vertices. The Tutte polynomial on half-edged graphs satisfies a recurrence rule of the form, for any regular edge $e$, \cite{avo}:
\beq
T_{G} =  T_{G\vee e} + T_{G/e} = t_e^{2}T_{G-e} + T_{G/e}\,, 
\eeq
where $t_e$ is a new variable associated with half-edges obtained
from $e$ after the cut in the graph $G$. One can quickly realize that this is an evaluation of 
the usual Tutte polynomial by its universality theorem. 
What it is important now to mention is the following: 
the half-edges let attached to the graph after cutting 
edges might play precisely the role of a memory of the system.
Indeed, whenever we cut an edge, the data of that edge is not totally
removed from the graph and might be encoded using the presence of
the half-edges. This track deserves to be further investigated.

\bigskip 

\section{Concluding remarks and perspectives}\label{sec:conc}
In this work, motivated by the recursion relation appearing in evaluating the moments of the integrated geometric Brownian motion, we have provided an extension of the Tutte-Fortuin-Kasteleyn polynomial depending on the contraction-deletion ordering. Such dependence emerges naturally from the recursion property proved in \cite{Caravelli1} for generic moments of the integrated geometric Brownian motion. 
Our work entails several interesting consequences. 
First, we have provided a possible generalization of  Tutte polynomial which reduces to the Tutte-FK polynomial in a specific limit of the control parameters. Moreover, we have provided sufficient evidence motivating the study of Tutte polynomials which depend on the ordering of contractions. We have suggested that this generalization might be based on a noncommutative version of the Tutte polynomial, arising from the nonassociative properties of a star product emerging from the ordering of the contraction-deletions. As it is well known, the Tutte polynomial is associated  to  the partition function of the Potts model \cite{Sokal}, and in the context of Quantum Field Theory \cite{Krajewski:2008fa},
the Tutte-Symanzik polynomial in fact provides a  parametric representation of Feynman amplitudes.
Thus, quite intriguingly, our work might relate indirectly stochastic models in one dimensions \cite{Caravelli1,Caravelli2} with the perturbation theory of Quantum Field Theories.  One of the main perspectives of the combinatorial approaches presented here is their extension for the study of the properties of other stochastic models which can be formally written as a perturbation expansion involving the integrated geometric Brownian motion.

\appendix

\renewcommand{\theequation}{\Alph{section}.\arabic{equation}}
\setcounter{equation}{0}

\section{ Proof of Lemma \ref{lemCkl} }
\label{app:lemckl}

Let us recall Lemma \ref{lemCkl}: 
\begin{lemma}
In notations of Theorem \ref{thm1}, for all $k > 1$, 
\beq\label{bin}
C_{k1;B}(G;\eps,\eps') = \eps^{1}_B \sum_{\ell>1}^kA^{(0)}_{\ell 1}  C_{k\ell;B}(G;\eps,\eps') \,.
\eeq
\end{lemma}

Working at fixed graphs $B$ and $G$, we simplify the notations as follows: $C_{k\ell;B}(G;\eps,\eps')=C_{k\ell}$.   We prove the relation \eqref{bin} by recurrence on $k$. At $k=2$, we have
$C_{21}= \eps^{1}_B A^{(0)}_{21}
 = \eps^{1}_B A^{(0)}_{21}  C_{22}$. Let us assume the
result holds at $k$. 
Calculating the r.h.s of \eqref{bin} at the next order $k+1$, we find  
\bea\label{split}
&&
\eps^{1}_B \sum_{\ell>1}^{k+1}A^{(0)}_{\ell1}  C_{(k+1)\ell}  = 
\eps^{1}_BA^{(0)}_{(k+1)1}  + \eps^{1}_BA^{(0)}_{k1}C_{(k+1)k} 
+ \eps^{1}_BA^{(0)}_{(k-1)1}C_{(k+1)(k-1)}\\
&&+ 
\eps^{1}_B \sum_{\ell>1}^{k-2}A^{(0)}_{\ell1}  
 \sum_{p=0}^{(k+1)-\ell-1}
\sum_{\stackrel{Q_p \subseteq \{\ell +1,\dots, (k+1)-1\}}{ |Q_p|=p }}
\mathcal{A}_{(k+1)\ell;B} (Q_p) \cr\cr\cr
&&=\eps^{1}_BA^{(0)}_{(k+1)1}   + \eps^{1}_BA^{(0)}_{k1}C_{(k+1)k} + \eps^{1}_BA^{(0)}_{(k-1)1}C_{(k+1)(k-1)}\crcr
&&+ 
\eps^{1}_B \sum_{\ell>1}^{k-2}A^{(0)}_{\ell1}  
 \sum_{p=0}^{k-\ell}
\Big[ 
\sum_{\stackrel{Q_p \subseteq \{\ell +1,\dots, k\}}{ |Q_p|=p \;\text{and}\; k\in Q_p }}  + 
\sum_{\stackrel{Q_p \subseteq \{\ell +1,\dots, k\}}{ |Q_p|=p \; \text{and}\; k\notin Q_p }}  \Big]  
\mathcal{A}_{(k+1)\ell;B} (Q_p) \cr\cr\cr
&&
=\eps^{1}_BA^{(0)}_{(k+1)1}   + \eps^{1}_BA^{(0)}_{k1}C_{(k+1)k} + \eps^{1}_BA^{(0)}_{(k-1)1}C_{(k+1)(k-1)}\crcr
&&+ 
\eps^{1}_B \sum_{\ell>1}^{k-2}A^{(0)}_{\ell1}  
 \sum_{p=0}^{k-\ell}
\Big[ 
\sum_{\stackrel{Q_p \subseteq \{\ell +1,\dots, k\}}{ |Q_p|=p \;\text{and}\; k\in Q_p }}  + 
\sum_{\stackrel{Q_p \subseteq \{\ell +1,\dots, k-1\}}{ |Q_p|=p }}  \Big]  
\mathcal{A}_{(k+1)\ell;B} (Q_p) \,. 
\nonumber
\eea
We expand  $C_{(k+1)(k-1)}$ as
\beq
C_{(k+1)(k-1)} = \mathcal{A}_{(k+1)(k-1)}(Q_0)+ 
\mathcal{A}_{(k+1)(k-1)}(\{k\}) = \eps^{k-1}_B (A^{(k-2)}_{(k+1)(k-1)}
+ \eps^{k}_B A^{(k-1)}_{(k+1)k}A^{(k-2)}_{k(k-1)})\,.
\eeq 
For the terms in  \eqref{split} such that $k \in Q_p$, which implies $p\geq 1$, because $j_{p}^Q\leq k $, we must have $j^Q_p=k$ and so we can write
\beq\label{extr}
\mathcal{A}_{(k+1)\ell;B} (Q_p)  = 
\eps_B^{k}A^{(k-1)}_{(k+1) k}  \Big(\eps^\ell_B \prod_{a=1}^{p-1} \eps^{j^Q_a}_B \Big) 
\Big[A^{(j^Q_{p-1}-1)}_{k j^Q_{p-1}}\Big(\prod_{i=2}^{p-1} A^{(j^Q_{i-1}-1)}_{j^Q_i j^Q_{i-1}} \Big)A^{(\ell-1)}_{ j^Q_1 \ell} \Big].
\eeq
The sum over the terms in \eqref{split} with $k\in Q_{p\geq 1}$ yields, after adjusting the variable $p-1 \to p$, 
\bea
&&   \eps^{1}_BA^{(0)}_{k1}C_{(k+1)k}  +  \eps^{k}_B\eps^{1}_BA^{(0)}_{(k-1)1}  A^{(k-1)}_{(k+1)k}A^{(k-2)}_{k(k-1)} \crcr
&&+ 
\eps^{1}_B \sum_{\ell>1}^{k-2}A^{(0)}_{\ell1}  
 \sum_{p=1}^{k-\ell}
\Big[ 
\sum_{\stackrel{Q_p \subseteq \{\ell +1,\dots, k\}}{ |Q_p|=p \;\text{and}\; k\in Q_p }}  \Big]  
\eps_B^{k}A^{(k-1)}_{(k+1) k}  \Big(\eps^\ell_B \prod_{a=1}^{p-1} \eps^{j^Q_a}_B \Big) 
\Big[A^{(j^Q_{p-1}-1)}_{k j^Q_{p-1}}\Big(\prod_{i=2}^{p-1} A^{(j^Q_{i-1}-1)}_{j^Q_i j^Q_{i-1}} \Big)A^{(\ell-1)}_{ j^Q_1 \ell} \Big]
\cr\cr\cr
&& =  \eps^{1}_BA^{(0)}_{k1} \eps^{k}_B A^{(k-1)}_{(k+1)k}   +  \eps^{k}_B\eps^{1}_BA^{(0)}_{(k-1)1} A^{(k-1)}_{(k+1)k}A^{(k-2)}_{k(k-1)} \cr\cr
&&+ 
 \eps_B^{k}A^{(k-1)}_{(k+1) k}\;  \eps^{1}_B \sum_{\ell>1}^{k-2}A^{(0)}_{\ell1}  
 \sum_{p=1}^{k-\ell}
\Big[ 
\sum_{\stackrel{Q_p \subseteq \{\ell +1,\dots, k-1\}}{ |Q_p|=p -1 }}  \Big]  
 \Big(\eps^\ell_B \prod_{a=1}^{p-1} \eps^{j^Q_a}_B \Big) 
\Big[A^{(j^Q_{p-1}-1)}_{k j^Q_{p-1}}\Big(\prod_{i=2}^{p-1} A^{(j^Q_{i-1}-1)}_{j^Q_i j^Q_{i-1}} \Big)A^{(\ell-1)}_{ j^Q_1 \ell} \Big]
\cr\cr
&& =  \eps^{1}_BA^{(0)}_{k1} \eps^{k}_B A^{(k-1)}_{(k+1)k}   +  \eps^{k}_B \eps^{1}_BA^{(0)}_{(k-1)1} A^{(k-1)}_{(k+1)k}A^{(k-2)}_{k(k-1)}  \cr\cr
&&+
 \eps_B^{k}A^{(k-1)}_{(k+1) k}\;  \eps^{1}_B \sum_{\ell>1}^{k-2}A^{(0)}_{\ell1}  
 \sum_{p=0}^{k-\ell-1}
\sum_{\stackrel{Q_p \subseteq \{\ell +1,\dots, k-1\}}{ |Q_p|=p }}   
 \mathcal{A}_{k\ell;B} (Q_p) 
 = \eps_B^{k}A^{(k-1)}_{(k+1) k}\;  \eps^{1}_B \sum_{l>1}^{k}A^{(0)}_{\ell 1}  
C_{k\ell}  \crcr
&& 
=\eps_B^{k}A^{(k-1)}_{(k+1) k}\;   C_{k1}\,,
\label{init5}
\eea
where in the last line, use has been made of the recurrence hypothesis. 
Let us concentrate on the second type of terms $k\notin Q_p$ in
 \eqref{split} that we
write, because there is no subsets of size $k-\ell>0$ (omitting
$\eps^{1}_BA^{(0)}_{(k+1)1}$)
\bea
&&
 \eps^{k-1}_B  \eps^{1}_BA^{(0)}_{(k-1)1} A^{(k-2)}_{(k+1)(k-1)} + 
\eps^{1}_B \sum_{\ell >1}^{k-2}A^{(0)}_{\ell 1}  
 \sum_{p=0}^{k-\ell-1}
\Big[ 
\sum_{\stackrel{Q_p \subseteq \{\ell +1,\dots, k-1\}}{ |Q_p|=p }}  \Big]  
\mathcal{A}_{(k+1)\ell;B} (Q_p) \cr\cr
&&
=  \eps^{k-1}_B  \eps^{1}_BA^{(0)}_{(k-1)1} A^{(k-2)}_{(k+1)(k-1)} + 
\eps^{1}_B A^{(0)}_{(k-2) 1}  (\mathcal{A}_{(k+1)(k-2);B} (Q_0)+
 \mathcal{A}_{(k+1)(k-2);B} (\{k-1\}))\cr\cr
&& 
+ 
\eps^{1}_B \sum_{\ell >1}^{k-3}A^{(0)}_{\ell 1}  
 \sum_{p=0}^{k-\ell-1}
\Big[ 
\sum_{\stackrel{Q_p \subseteq \{\ell +1,\dots, k-1\}}{ |Q_p|=p }}  \Big]  
\mathcal{A}_{(k+1)\ell;B} (Q_p) 
\cr\cr
&&
=  \eps^{k-1}_B  \eps^{1}_BA^{(0)}_{(k-1)1} A^{(k-2)}_{(k+1)(k-1)}
 + 
\eps^{1}_B A^{(0)}_{(k-2) 1}  (\eps_B^{k-2} A^{(k-1)}_{(k+1)(k-2)} +
\eps^{k-2}_B \eps^{k-1}_B A^{(k-2)}_{(k+1)(k-1)} A^{(k-3)}_{(k-1)(k-2)} )
\cr\cr
&& 
+ 
\eps^{1}_B \sum_{\ell >1}^{k-3}A^{(0)}_{\ell 1}  
 \sum_{p=0}^{k-\ell-1}
\Big[ 
\sum_{\stackrel{Q_p \subseteq \{\ell +1,\dots, k-1\}}{ |Q_p|=p }}  \Big]  
\mathcal{A}_{(k+1)\ell;B} (Q_p)   \,. 
\label{rest}
\eea
Now we iterate the same procedure and split the sum over subsets
$Q_p \subset\{\ell+1,\dots, k-1\}$, among the terms containing
$k-1$ and those which do not. Using the same decomposition as
in \eqref{extr} followed by the expansion \eqref{init5}, we obtain
by summing those terms containing $k-1$, and a few algebra  from \eqref{rest}:
\bea
&& \eps^{k-1}_B A^{(k-2)}_{(k+1)(k-1)}\Big[ \eps^{1}_BA^{(0)}_{(k-1)1}
 +
\eps^{1}_B\eps^{k-2}_B   A^{(0)}_{(k-2) 1}  A^{(k-3)}_{(k-1)(k-2)} 
 +
\eps^{1}_B \sum_{\ell >1}^{k-3}A^{(0)}_{\ell 1}
\sum_{p=0}^{k-\ell -2}  \sum_{\stackrel{Q_p \subseteq \{\ell +1,\dots, k-2\}}{ |Q_p|=p }} \mathcal{A}_{(k-1)\ell;B} \Big] 
\cr\cr
&&
= 
 \eps^{k-1}_B A^{(k-2)}_{(k+1)(k-1)}\Big[
 \eps^{1}_BA^{(0)}_{(k-1)1} C_{(k-1)(k-1)}
+ 
 \eps^{1}_BA^{(0)}_{(k-2)1} C_{(k-1)(k-2)}
+
\eps^{1}_B \sum_{\ell >1}^{k-3}A^{(0)}_{\ell 1}C_{(k-1)\ell}
 \Big] 
\cr\cr
&&
= 
 \eps^{k-1}_B A^{(k-2)}_{(k+1)(k-1)}\Big[
\eps^{1}_B \sum_{\ell>1}^{k-1}
 A^{(0)}_{\ell 1} C_{(k-1)\ell}
 \Big] 
= \eps^{k-1}_B A^{(k-2)}_{(k+1)(k-1)} C_{(k-1)1}\,,
\eea
where again we have used the recurrence hypothesis,  this
time, at order $k-1$. The procedure can be pursued until no more
terms are left in the sum, and one gets as an upshot
\bea
&&
\eps^{1}_B \sum_{\ell>1}^{k+1}A^{(0)}_{\ell1}  C_{(k+1)\ell}   = 
\eps_B^{k}A^{(k-1)}_{(k+1) k}\;   C_{k1} 
 +  \eps^{k-1}_B A^{(k-2)}_{(k+1)(k-1)} C_{(k-1)1}
+ \dots +  
 \eps^{2}_B A^{(1)}_{(k+1)2} C_{21}
+\eps^{1}_BA^{(0)}_{(k+1)1} C_{11} \cr\cr
&&
=\sum_{\ell=1}^{k} \eps_B^{\ell }A^{(\ell -1)}_{(k+1) \ell }\;   C_{\ell 1} \,.
\eea
We substitute the value of $C_{\ell 1}$ in the above expression 
and obtain
\bea
&&
\eps^{1}_B \sum_{\ell>1}^{k+1}A^{(0)}_{\ell1}  C_{(k+1)\ell}   
= \eps^{1}_BA^{(0)}_{(k+1)1}
 + \sum_{\ell=2}^{k} \; 
\sum_{p=0}^{\ell-2}
\sum_{\stackrel{Q_p \subseteq \{2,\dots, \ell -1\}}{ |Q_p|=p }}
\eps_B^{\ell }A^{(\ell -1)}_{(k+1) \ell }\, 
\mathcal{A}_{\ell1;B} (Q_p) \cr\cr
&&
= \eps^{1}_BA^{(0)}_{(k+1)1}
 +  \sum_{p=0}^{k-2}
\sum_{\stackrel{Q_p \subseteq \{2,\dots, k -1\}}{ |Q_p|=p }}
\sum_{\ell=\max(p+2,j^Q_{p}+1)}^{k} \eps_B^{\ell }A^{(\ell -1)}_{(k+1) \ell }\, 
\mathcal{A}_{\ell1;B} (Q_p) \,.
\eea
Observing that $2\leq j^Q_1< \dots < j^Q_p $, then 
$j^Q_p \geq p+2$ such that $\max(p+2,j^Q_{p}+1)=j^Q_{p}+1$. 
Furthermore, because $j^Q_p < \ell $, we re-express the above formula as
\bea
&&
\eps^{1}_B \sum_{\ell>1}^{k+1}A^{(0)}_{\ell1}  C_{(k+1)\ell} =  
\eps^{1}_BA^{(0)}_{(k+1)1}
 +  \sum_{p=0}^{k-2}
\sum_{\stackrel{Q_p \subseteq \{2,\dots, k -1\}}{ |Q_p|=p }}
\sum_{\ell=j^Q_{p}+1}^{k}  
\mathcal{A}_{(k+1)1;B} (Q_p \cup \{\ell\}) \cr\cr
&&= 
\eps^{1}_BA^{(0)}_{(k+1)1}
 +  \sum_{p=0}^{k-2}
\sum_{\stackrel{Q_{p+1} \subseteq \{2,\dots, k\}}{ |Q_{p+1}|=p+1 }} 
\mathcal{A}_{(k+1)1;B} (Q_{p+1}) = 
\eps^{1}_BA^{(0)}_{(k+1)1}
 +  \sum_{p=1}^{k-1}
\sum_{\stackrel{Q_{p} \subseteq \{2,\dots, k\}}{ |Q_{p}|=p }} 
\mathcal{A}_{(k+1)1;B} (Q_{p}) \cr\cr
&& 
 =  \sum_{p=0}^{k-1}
\sum_{\stackrel{Q_{p} \subseteq \{2,\dots, k\}}{ |Q_{p}|=p }} 
\mathcal{A}_{(k+1)1;B} (Q_{p}) = C_{(k+1)1}. 
\eea
\qed

\section{A worked out example}\label{app:example}

Consider the graph $G$ of Fig.\ref{fig:ex}, with edge set 
$\{e_{i}\}_{i=1,\dots, 5}$. We call $\{e_1,e_3,e_5\}=E_1$ 
the edge set of $G_1$, and $\{e_2,e_4\} = E_2$, the edge set of $G_2$. 
Note that the subgraphs $G_1$ and $G_2$ are symmetric under relabelling, so it does not matter to claim which edge is which 
in that figure. 

Let us pick the sequence $(e_1,e_2,e_3,e_4,e_5)$ in which 
we want to perform the contraction deletion. 
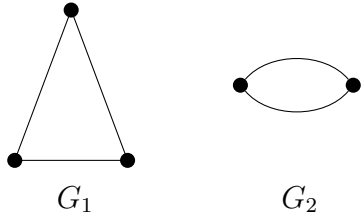
\begin{figure}[h]
\center
\begin{tikzpicture}
\draw (0,0) -- (1.5,0) -- (0.75,2) -- (0,0) ;

 \fill[fill=black] (0,0) circle (0.1);
 \fill[fill=black] (1.5,0) circle (0.1);
 \fill[fill=black] (0.75,2) circle (0.1);

\fill[fill=black] (3,1) circle (0.1);
\fill[fill=black] (4.5,1) circle (0.1);

 \draw[shorten <= 0.1cm, shorten >= 0.1cm] (3,1) to[out=50, in=130] (4.5,1);
\draw[shorten <= 0.1cm, shorten >= 0.1cm] (3,1) to[out=-50, in=-130] (4.5,1);

\end{tikzpicture}
\put(-115,-15){$G_1$}
\put(-30,-15){$G_2$}
\caption{A disconnected graph $G=G_1 \cup G_2$} 
\label{fig:ex}
\end{figure}
Now, we can apply the state sum to find the polynomial. 
We will use a graphical representation to compute the 
$2^5$ terms of the sum and will simplify the notations. We have 
\bea
&& 
P^{\eps,\eps'}(G;q;\{\lambda_j\})
 = \sum_{B \subset G} q^{k(B)} \; \prod_{j=1}^n \gamma_B (\widehat\lambda_{j} \big(\{\lambda_l\}_{l\leq j}) \big)  \cr\cr
&&= 
q^{2}\prod_{j=1}^{5}
 \gamma_{\begin{tikzpicture}
\draw (0,0) -- (1.5/4,0) -- (0.75/4,2/4) -- (0,0) ;
 \fill[fill=black] (0,0) circle (0.1/4);
 \fill[fill=black] (1.5/4,0) circle (0.1/4);
 \fill[fill=black] (0.75/4,2/4) circle (0.1/4);
\fill[fill=black] (3/4,1/4) circle (0.1/4);
\fill[fill=black] (4.5/4,1/4) circle (0.1/4);
 \draw[shorten <= 0.1cm, shorten >= 0.1cm] (3/4,1/4) to[out=50, in=130] (4.5/4,1/4);
\draw[shorten <= 0.1cm, shorten >= 0.1cm] (3/4,1/4) to[out=-50, in=-130] (4.5/4,1/4);
\end{tikzpicture}} 
(\widehat\lambda_{j} \big(\{\lambda_l\}_{l\leq j}) \big) 
 + q^{2}\prod_{j=1}^{5}
 \gamma_{\begin{tikzpicture}
\draw (0,0) -- (1.5/4,0) -- (0.75/4,2/4) -- (0,0) ;
 \fill[fill=black] (0,0) circle (0.1/4);
 \fill[fill=black] (1.5/4,0) circle (0.1/4);
 \fill[fill=black] (0.75/4,2/4) circle (0.1/4);
\fill[fill=black] (3/4,1/4) circle (0.1/4);
\fill[fill=black] (4.5/4,1/4) circle (0.1/4);
\draw[shorten <= 0.1cm, shorten >= 0.1cm] (3/4,1/4) to[out=-50, in=-130] (4.5/4,1/4);
\end{tikzpicture}} \;\; 
(\widehat\lambda_{j} \big(\{\lambda_l\}_{l\leq j}) \big) \cr\cr
&&
+  
 q^{2}\prod_{j=1}^{5}
 \gamma_{\begin{tikzpicture}
\draw (0,0) -- (1.5/4,0) -- (0.75/4,2/4) -- (0,0) ;
 \fill[fill=black] (0,0) circle (0.1/4);
 \fill[fill=black] (1.5/4,0) circle (0.1/4);
 \fill[fill=black] (0.75/4,2/4) circle (0.1/4);
\fill[fill=black] (3/4,1/4) circle (0.1/4);
\fill[fill=black] (4.5/4,1/4) circle (0.1/4);
 \draw[shorten <= 0.1cm, shorten >= 0.1cm] (3/4,1/4) to[out=50, in=130] (4.5/4,1/4);
\end{tikzpicture}} 
(\widehat\lambda_{j} \big(\{\lambda_l\}_{l\leq j}) \big) 
 + q^{3}\prod_{j=1}^{5}
 \gamma_{\begin{tikzpicture}
\draw (0,0) -- (1.5/4,0) -- (0.75/4,2/4) -- (0,0) ;
 \fill[fill=black] (0,0) circle (0.1/4);
 \fill[fill=black] (1.5/4,0) circle (0.1/4);
 \fill[fill=black] (0.75/4,2/4) circle (0.1/4);
\fill[fill=black] (3/4,1/4) circle (0.1/4);
\fill[fill=black] (4.5/4,1/4) circle (0.1/4);
\end{tikzpicture}} \;\; 
(\widehat\lambda_{j} \big(\{\lambda_l\}_{l\leq j}) \big) 
\cr\cr
&&
+ q^{2}\prod_{j=1}^{5}
 \gamma_{\begin{tikzpicture}
\draw (0,0) -- (1.5/4,0) -- (0.75/4,2/4) ; 
 \fill[fill=black] (1.5/4,0) circle (0.1/4);
 \fill[fill=black] (0.75/4,2/4) circle (0.1/4);
\fill[fill=black] (3/4,1/4) circle (0.1/4);
\fill[fill=black] (4.5/4,1/4) circle (0.1/4);
 \draw[shorten <= 0.1cm, shorten >= 0.1cm] (3/4,1/4) to[out=50, in=130] (4.5/4,1/4);
\draw[shorten <= 0.1cm, shorten >= 0.1cm] (3/4,1/4) to[out=-50, in=-130] (4.5/4,1/4);
\end{tikzpicture}} 
(\widehat\lambda_{j} \big(\{\lambda_l\}_{l\leq j}) \big) 
 + q^{2}\prod_{j=1}^{5}
 \gamma_{\begin{tikzpicture}
\draw (0,0) -- (1.5/4,0) -- (0.75/4,2/4);  
 \fill[fill=black] (0,0) circle (0.1/4);
 \fill[fill=black] (1.5/4,0) circle (0.1/4);
 \fill[fill=black] (0.75/4,2/4) circle (0.1/4);
\fill[fill=black] (3/4,1/4) circle (0.1/4);
\fill[fill=black] (4.5/4,1/4) circle (0.1/4);
\draw[shorten <= 0.1cm, shorten >= 0.1cm] (3/4,1/4) to[out=-50, in=-130] (4.5/4,1/4);
\end{tikzpicture}} \;\; 
(\widehat\lambda_{j} \big(\{\lambda_l\}_{l\leq j}) \big) \cr\cr
&&
+  
 q^{2}\prod_{j=1}^{5}
 \gamma_{\begin{tikzpicture}
\draw (0,0) -- (1.5/4,0) -- (0.75/4,2/4); 
 \fill[fill=black] (0,0) circle (0.1/4);
 \fill[fill=black] (1.5/4,0) circle (0.1/4);
 \fill[fill=black] (0.75/4,2/4) circle (0.1/4);
\fill[fill=black] (3/4,1/4) circle (0.1/4);
\fill[fill=black] (4.5/4,1/4) circle (0.1/4);
 \draw[shorten <= 0.1cm, shorten >= 0.1cm] (3/4,1/4) to[out=50, in=130] (4.5/4,1/4);
\end{tikzpicture}} 
(\widehat\lambda_{j} \big(\{\lambda_l\}_{l\leq j}) \big) 
 + q^{3}\prod_{j=1}^{5}
 \gamma_{\begin{tikzpicture}
\draw (0,0) -- (1.5/4,0) -- (0.75/4,2/4); 
 \fill[fill=black] (0,0) circle (0.1/4);
 \fill[fill=black] (1.5/4,0) circle (0.1/4);
 \fill[fill=black] (0.75/4,2/4) circle (0.1/4);
\fill[fill=black] (3/4,1/4) circle (0.1/4);
\fill[fill=black] (4.5/4,1/4) circle (0.1/4);
\end{tikzpicture}} \;\; 
(\widehat\lambda_{j} \big(\{\lambda_l\}_{l\leq j}) \big) 
\cr\cr
&& 
+\Bigg\{ {\mbox{and the like obtained by replacing  
\begin{tikzpicture}
\draw (0,0) -- (1.5/4,0) -- (0.75/4,2/4);
 \fill[fill=black] (0,0) circle (0.1/4);
 \fill[fill=black] (1.5/4,0) circle (0.1/4);
 \fill[fill=black] (0.75/4,2/4) circle (0.1/4);
\end{tikzpicture} by  \begin{tikzpicture}
\draw (1.5/4,0) -- (0.75/4,2/4) -- (0,0) ;
 \fill[fill=black] (0,0) circle (0.1/4);
 \fill[fill=black] (1.5/4,0) circle (0.1/4);
 \fill[fill=black] (0.75/4,2/4) circle (0.1/4);
\end{tikzpicture} 
,  \begin{tikzpicture}
\draw (1.5/4,0) -- (0,0) --  (0.75/4,2/4);
 \fill[fill=black] (0,0) circle (0.1/4);
 \fill[fill=black] (1.5/4,0) circle (0.1/4);
 \fill[fill=black] (0.75/4,2/4) circle (0.1/4);
\end{tikzpicture}  }} \Bigg\}   \cr\cr
&&
+ q^{3}\prod_{j=1}^{5}
 \gamma_{\begin{tikzpicture}
\draw (0,0) -- (1.5/4,0);   
 \fill[fill=black] (0,0) circle (0.1/4);
 \fill[fill=black] (1.5/4,0) circle (0.1/4);
 \fill[fill=black] (0.75/4,2/4) circle (0.1/4);
\fill[fill=black] (3/4,1/4) circle (0.1/4);
\fill[fill=black] (4.5/4,1/4) circle (0.1/4);
 \draw[shorten <= 0.1cm, shorten >= 0.1cm] (3/4,1/4) to[out=50, in=130] (4.5/4,1/4);
\draw[shorten <= 0.1cm, shorten >= 0.1cm] (3/4,1/4) to[out=-50, in=-130] (4.5/4,1/4);
\end{tikzpicture}} 
(\widehat\lambda_{j} \big(\{\lambda_l\}_{l\leq j}) \big) 
 + q^{3}\prod_{j=1}^{5}
 \gamma_{\begin{tikzpicture}
\draw (0,0) -- (1.5/4,0); 
 \fill[fill=black] (0,0) circle (0.1/4);
 \fill[fill=black] (1.5/4,0) circle (0.1/4);
 \fill[fill=black] (0.75/4,2/4) circle (0.1/4);
\fill[fill=black] (3/4,1/4) circle (0.1/4);
\fill[fill=black] (4.5/4,1/4) circle (0.1/4);
\draw[shorten <= 0.1cm, shorten >= 0.1cm] (3/4,1/4) to[out=-50, in=-130] (4.5/4,1/4);
\end{tikzpicture}} \;\; 
(\widehat\lambda_{j} \big(\{\lambda_l\}_{l\leq j}) \big) \cr\cr
&&
+  
 q^{3}\prod_{j=1}^{5}
 \gamma_{\begin{tikzpicture}
\draw (0,0) -- (1.5/4,0);   
 \fill[fill=black] (0,0) circle (0.1/4);
 \fill[fill=black] (1.5/4,0) circle (0.1/4);
 \fill[fill=black] (0.75/4,2/4) circle (0.1/4);
\fill[fill=black] (3/4,1/4) circle (0.1/4);
\fill[fill=black] (4.5/4,1/4) circle (0.1/4);
 \draw[shorten <= 0.1cm, shorten >= 0.1cm] (3/4,1/4) to[out=50, in=130] (4.5/4,1/4);
\end{tikzpicture}} 
(\widehat\lambda_{j} \big(\{\lambda_l\}_{l\leq j}) \big) 
 + q^{4}\prod_{j=1}^{5}
 \gamma_{\begin{tikzpicture}
\draw (0,0) -- (1.5/4,0);  
 \fill[fill=black] (0,0) circle (0.1/4);
 \fill[fill=black] (1.5/4,0) circle (0.1/4);
 \fill[fill=black] (0.75/4,2/4) circle (0.1/4);
\fill[fill=black] (3/4,1/4) circle (0.1/4);
\fill[fill=black] (4.5/4,1/4) circle (0.1/4);
\end{tikzpicture}} \;\; 
(\widehat\lambda_{j} \big(\{\lambda_l\}_{l\leq j}) \big) 
\cr\cr
&& 
+ \Big\{ {\mbox{and the like obtained by replacing  \begin{tikzpicture}
\draw (0,0) -- (1.5/4,0);   
 \fill[fill=black] (0,0) circle (0.1/4);
 \fill[fill=black] (1.5/4,0) circle (0.1/4);
 \fill[fill=black] (0.75/4,2/4) circle (0.1/4);
\end{tikzpicture} by  \begin{tikzpicture}
\draw (0.75/4,2/4) -- (0,0) ;
 \fill[fill=black] (0,0) circle (0.1/4);
 \fill[fill=black] (1.5/4,0) circle (0.1/4);
 \fill[fill=black] (0.75/4,2/4) circle (0.1/4);
\end{tikzpicture} 
,  \begin{tikzpicture}
\draw (1.5/4,0) --  (0.75/4,2/4);
 \fill[fill=black] (0,0) circle (0.1/4);
 \fill[fill=black] (1.5/4,0) circle (0.1/4);
 \fill[fill=black] (0.75/4,2/4) circle (0.1/4);
\end{tikzpicture}  }} \Big\} \cr\cr
&& 
+ q^{4}\prod_{j=1}^{5}
 \gamma_{\begin{tikzpicture}
 \fill[fill=black] (0,0) circle (0.1/4);
 \fill[fill=black] (1.5/4,0) circle (0.1/4);
 \fill[fill=black] (0.75/4,2/4) circle (0.1/4);
\fill[fill=black] (3/4,1/4) circle (0.1/4);
\fill[fill=black] (4.5/4,1/4) circle (0.1/4);
 \draw[shorten <= 0.1cm, shorten >= 0.1cm] (3/4,1/4) to[out=50, in=130] (4.5/4,1/4);
\draw[shorten <= 0.1cm, shorten >= 0.1cm] (3/4,1/4) to[out=-50, in=-130] (4.5/4,1/4);
\end{tikzpicture}} 
(\widehat\lambda_{j} \big(\{\lambda_l\}_{l\leq j}) \big) 
 + q^{4}\prod_{j=1}^{5}
 \gamma_{\begin{tikzpicture}
 \fill[fill=black] (0,0) circle (0.1/4);
 \fill[fill=black] (1.5/4,0) circle (0.1/4);
 \fill[fill=black] (0.75/4,2/4) circle (0.1/4);
\fill[fill=black] (3/4,1/4) circle (0.1/4);
\fill[fill=black] (4.5/4,1/4) circle (0.1/4);
\draw[shorten <= 0.1cm, shorten >= 0.1cm] (3/4,1/4) to[out=-50, in=-130] (4.5/4,1/4);
\end{tikzpicture}} \;\; 
(\widehat\lambda_{j} \big(\{\lambda_l\}_{l\leq j}) \big) \cr\cr
&&
+  
 q^{4}\prod_{j=1}^{5}
 \gamma_{\begin{tikzpicture}
 \fill[fill=black] (0,0) circle (0.1/4);
 \fill[fill=black] (1.5/4,0) circle (0.1/4);
 \fill[fill=black] (0.75/4,2/4) circle (0.1/4);
\fill[fill=black] (3/4,1/4) circle (0.1/4);
\fill[fill=black] (4.5/4,1/4) circle (0.1/4);
 \draw[shorten <= 0.1cm, shorten >= 0.1cm] (3/4,1/4) to[out=50, in=130] (4.5/4,1/4);
\end{tikzpicture}} 
(\widehat\lambda_{j} \big(\{\lambda_l\}_{l\leq j}) \big) 
 + q^{5}\prod_{j=1}^{5}
 \gamma_{\begin{tikzpicture}
 \fill[fill=black] (0,0) circle (0.1/4);
 \fill[fill=black] (1.5/4,0) circle (0.1/4);
 \fill[fill=black] (0.75/4,2/4) circle (0.1/4);
\fill[fill=black] (3/4,1/4) circle (0.1/4);
\fill[fill=black] (4.5/4,1/4) circle (0.1/4);
\end{tikzpicture}} \;\; 
(\widehat\lambda_{j} \big(\{\lambda_l\}_{l\leq j}) \big) \,.
\eea
Let us concentrate on the first term. Because the subgraph 
$B=\begin{tikzpicture}
\draw (0,0) -- (1.5/4,0) -- (0.75/4,2/4) -- (0,0) ;
 \fill[fill=black] (0,0) circle (0.1/4);
 \fill[fill=black] (1.5/4,0) circle (0.1/4);
 \fill[fill=black] (0.75/4,2/4) circle (0.1/4);
\fill[fill=black] (3/4,1/4) circle (0.1/4);
\fill[fill=black] (4.5/4,1/4) circle (0.1/4);
 \draw[shorten <= 0.1cm, shorten >= 0.1cm] (3/4,1/4) to[out=50, in=130] (4.5/4,1/4);
\draw[shorten <= 0.1cm, shorten >= 0.1cm] (3/4,1/4) to[out=-50, in=-130] (4.5/4,1/4);
\end{tikzpicture}$ contains all edges, we can write (working at fixed
subgraph $B$, we omit it in the notations): 
\bea
q^{2}\prod_{j=1}^{5}
 \gamma_{\begin{tikzpicture}
\draw (0,0) -- (1.5/4,0) -- (0.75/4,2/4) -- (0,0) ;
 \fill[fill=black] (0,0) circle (0.1/4);
 \fill[fill=black] (1.5/4,0) circle (0.1/4);
 \fill[fill=black] (0.75/4,2/4) circle (0.1/4);
\fill[fill=black] (3/4,1/4) circle (0.1/4);
\fill[fill=black] (4.5/4,1/4) circle (0.1/4);
 \draw[shorten <= 0.1cm, shorten >= 0.1cm] (3/4,1/4) to[out=50, in=130] (4.5/4,1/4);
\draw[shorten <= 0.1cm, shorten >= 0.1cm] (3/4,1/4) to[out=-50, in=-130] (4.5/4,1/4);
\end{tikzpicture}} 
(\widehat\lambda_{j} \big(\{\lambda_l\}_{l\leq j}) \big) 
 = q^{2}\prod_{j=1}^{5}
 \alpha_{\begin{tikzpicture}
\draw (0,0) -- (1.5/4,0) -- (0.75/4,2/4) -- (0,0) ;
 \fill[fill=black] (0,0) circle (0.1/4);
 \fill[fill=black] (1.5/4,0) circle (0.1/4);
 \fill[fill=black] (0.75/4,2/4) circle (0.1/4);
\fill[fill=black] (3/4,1/4) circle (0.1/4);
\fill[fill=black] (4.5/4,1/4) circle (0.1/4);
 \draw[shorten <= 0.1cm, shorten >= 0.1cm] (3/4,1/4) to[out=50, in=130] (4.5/4,1/4);
\draw[shorten <= 0.1cm, shorten >= 0.1cm] (3/4,1/4) to[out=-50, in=-130] (4.5/4,1/4);
\end{tikzpicture}} 
(\widehat\lambda_{j} \big(\{\lambda_l\}_{l\leq j}) \big) \,, 
\eea
with 
\bea
\widehat\lambda_{1}  &=& \lambda_1\,;
\cr\cr
\widehat\lambda_{2}&=&
 C_{2 1;B}(G;\eps,\eps')\lambda_1 +\lambda_2 = \eps \Big(A^{(0)}_{21} = 0
\Big) \lambda_1 +\lambda_2 =  \lambda_2 \,; 
\cr\cr
\widehat\lambda_{3}&=& 
C_{3 1}(G;\eps,\eps')\lambda_1 +C_{32}(G;\eps,\eps')\lambda_2
+ \lambda_3
=
 C_{3 1}(G;\eps,\eps')\lambda_1 +\eps \Big(A^{(1)}_{32} = 0\Big)\lambda_2 + \lambda_3 \cr\cr
&=&
\eps \Big(A^{(0)}_{31}=1\Big) \lambda_1 + \lambda_3 = \eps  \lambda_1 + \lambda_3\,, 
\cr\cr
\mbox{where}&&
C_{3 1}(G;\eps,\eps') = 
\mathcal{A}_{31} (\emptyset) +
\mathcal{A}_{31} (\{2\}) 
= 
\eps A^{(0)}_{31} + \eps^2\Big(A^{(1)}_{32}=0\Big)\Big(A^{(0)}_{21} 
= 0\Big) = \eps A^{(0)}_{31}\,; 
 \cr\cr 
\widehat\lambda_{4}&=& 
(C_{41}(G;\eps,\eps')=0)\lambda_1 
+ 
\eps (A^{(1)}_{42}=1) \lambda_2 
+
\eps (A^{0}_{43}=0)\lambda_3 
+ 
\lambda_4 =
 \eps  \lambda_2 + 
\lambda_4 \,, 
 \cr\cr
\mbox{where}&&
C_{41}(G;\eps,\eps') = 
\mathcal{A}_{41} (\emptyset) 
+ 
\mathcal{A}_{41} (\{2\}) 
+ 
\mathcal{A}_{41} (\{3\}) \cr\cr
&& = 
\eps (A^{(0)}_{41}=0) 
+ 
\eps^2 A^{(1)}_{42}(A^{(0)}_{21}=0) 
+ 
\eps^2 (A^{(2)}_{43}=0)A^{(0)}_{31} =0
\cr\cr
\mbox{and}&&
C_{42}(G;\eps,\eps') = 
\mathcal{A}_{42} (\emptyset) 
+ 
\mathcal{A}_{43} (\{3\})  = 
\eps A^{(1)}_{42}
+ 
\eps^2 (A^{(2)}_{43}=0)(A^{(1)}_{32} =0)= \eps A^{(1)}_{42}\,;
\cr\cr
\widehat\lambda_{5}
&=&
C_{51}(G;\eps,\eps')\lambda_1 
+ 
(C_{52}(G;\eps,\eps')=0)\lambda_2 
+
C_{53}(G;\eps,\eps')\lambda_3 
+ 
\eps (A^{(3)}_{54}=0)\lambda_4  
+ 
\lambda_5 \cr\cr
&=& 
(\eps + \eps^2)\lambda_1 
+
\eps  \lambda_3 
+ 
\lambda_5 \,,
\cr\cr
\mbox{where}&&
C_{51}(G;\eps,\eps') = 
\mathcal{A}_{51} (\emptyset) 
+  
\mathcal{A}_{51} (\{2,3,4\}) \crcr
&& 
+ 
\mathcal{A}_{51} (\{2\}) 
+
\mathcal{A}_{51} (\{3\}) +\mathcal{A}_{51} (\{4\})
 + 
\mathcal{A}_{51} (\{2,3\}) 
 + 
\mathcal{A}_{51} (\{2,4\}) 
+ 
\mathcal{A}_{51} (\{3,4\}) 
 \cr\cr
&& = \eps A^{(0)}_{51} + 
(\eps^{4}A^{(3)}_{54}...A^{(0)}_{21}=0)
+ \eps^2 (A^{(1)}_{52}A^{(0)}_{21}=0)
+ \eps^2 (A^{(2)}_{53}A^{(0)}_{31}=1)\cr\cr
&&
+ \eps^2 (A^{(4)}_{54}A^{(0)}_{41}=0) 
+ (\mathcal{A}_{51} (\{2,3\}) =0)
 + (\mathcal{A}_{51} (\{2,4\}) =0)
+ (\mathcal{A}_{51} (\{3,4\}) =0)
\cr\cr
&& =  \eps A^{(0)}_{51} + \eps^2 (A^{(2)}_{53}A^{(0)}_{31}=1)
 = \eps + \eps^2 \cr\cr 
\mbox{and}&&
C_{53}(G;\eps,\eps') = 
\mathcal{A}_{53} (\emptyset) 
+ 
(\mathcal{A}_{53} (\{4\}) =0) = \eps A^{(2)}_{53}  \,,
\eea
as
\bea
&& 
q^{2}\prod_{j=1}^{5}
 \gamma_{\begin{tikzpicture}
\draw (0,0) -- (1.5/4,0) -- (0.75/4,2/4) -- (0,0) ;
 \fill[fill=black] (0,0) circle (0.1/4);
 \fill[fill=black] (1.5/4,0) circle (0.1/4);
 \fill[fill=black] (0.75/4,2/4) circle (0.1/4);
\fill[fill=black] (3/4,1/4) circle (0.1/4);
\fill[fill=black] (4.5/4,1/4) circle (0.1/4);
 \draw[shorten <= 0.1cm, shorten >= 0.1cm] (3/4,1/4) to[out=50, in=130] (4.5/4,1/4);
\draw[shorten <= 0.1cm, shorten >= 0.1cm] (3/4,1/4) to[out=-50, in=-130] (4.5/4,1/4);
\end{tikzpicture}} 
(\widehat\lambda_{j} \big(\{\lambda_l\}_{l\leq j}) \big) 
 =\crcr
&&
 q^{2} \alpha(\lambda_1)\alpha(\lambda_2)
\alpha(\lambda_3 + \eps \lambda_1) 
\alpha(\lambda_4 + \eps \lambda_2)
\alpha(\lambda_5 + \eps \lambda_3 + \eps(1+\eps )\lambda_1)  
\cr\cr
&&
=\Big[q \alpha(\lambda_1)
\alpha(\lambda_3 + \eps \lambda_1) 
\alpha(\lambda_5 + \eps \lambda_3 + \eps(1+\eps )\lambda_1)  
\Big] 
\Big[ \alpha(\lambda_2)
\alpha(\lambda_4 + \eps \lambda_2)
\Big] 
\cr\cr
&&
= \big(q\prod_{j=1}^{5}
 \gamma_{\begin{tikzpicture}
\draw (0,0) -- (1.5/4,0) -- (0.75/4,2/4) -- (0,0) ;
 \fill[fill=black] (0,0) circle (0.1/4);
 \fill[fill=black] (1.5/4,0) circle (0.1/4);
 \fill[fill=black] (0.75/4,2/4) circle (0.1/4);
\end{tikzpicture}} \big)\big(q\prod_{j=1}^{5}
 \gamma_{\begin{tikzpicture}
\fill[fill=black] (3/4,1/4) circle (0.1/4);
\fill[fill=black] (4.5/4,1/4) circle (0.1/4);
 \draw[shorten <= 0.1cm, shorten >= 0.1cm] (3/4,1/4) to[out=50, in=130] (4.5/4,1/4);
\draw[shorten <= 0.1cm, shorten >= 0.1cm] (3/4,1/4) to[out=-50, in=-130] (4.5/4,1/4);
\end{tikzpicture}} \big)
\eea
where in the two last equalities, we emphasize the factorization
of this monomial between the contribution of two subgraphs 
of $G_1$ and $G_2$. 
Using the same technique, we list
\bea
&& 
q^{2}\prod_{j=1}^{5}
 \gamma_{\begin{tikzpicture}
\draw (0,0) -- (1.5/4,0) -- (0.75/4,2/4) -- (0,0) ;
 \fill[fill=black] (0,0) circle (0.1/4);
 \fill[fill=black] (1.5/4,0) circle (0.1/4);
 \fill[fill=black] (0.75/4,2/4) circle (0.1/4);
\fill[fill=black] (3/4,1/4) circle (0.1/4);
\fill[fill=black] (4.5/4,1/4) circle (0.1/4);
 \draw[shorten <= 0.1cm, shorten >= 0.1cm] (3/4,1/4) to[out=50, in=130] (4.5/4,1/4);
\draw[shorten <= 0.1cm, shorten >= 0.1cm] (3/4,1/4) to[out=-50, in=-130] (4.5/4,1/4);
\end{tikzpicture}} 
(\widehat\lambda_{j} \big(\{\lambda_l\}_{l\leq j}) \big) 
 + q^{2}\prod_{j=1}^{5}
 \gamma_{\begin{tikzpicture}
\draw (0,0) -- (1.5/4,0) -- (0.75/4,2/4) -- (0,0) ;
 \fill[fill=black] (0,0) circle (0.1/4);
 \fill[fill=black] (1.5/4,0) circle (0.1/4);
 \fill[fill=black] (0.75/4,2/4) circle (0.1/4);
\fill[fill=black] (3/4,1/4) circle (0.1/4);
\fill[fill=black] (4.5/4,1/4) circle (0.1/4);
\draw[shorten <= 0.1cm, shorten >= 0.1cm] (3/4,1/4) to[out=-50, in=-130] (4.5/4,1/4);
\end{tikzpicture}} \;\; 
(\widehat\lambda_{j} \big(\{\lambda_l\}_{l\leq j}) \big) \cr\cr
&&
+  
 q^{2}\prod_{j=1}^{5}
 \gamma_{\begin{tikzpicture}
\draw (0,0) -- (1.5/4,0) -- (0.75/4,2/4) -- (0,0) ;
 \fill[fill=black] (0,0) circle (0.1/4);
 \fill[fill=black] (1.5/4,0) circle (0.1/4);
 \fill[fill=black] (0.75/4,2/4) circle (0.1/4);
\fill[fill=black] (3/4,1/4) circle (0.1/4);
\fill[fill=black] (4.5/4,1/4) circle (0.1/4);
 \draw[shorten <= 0.1cm, shorten >= 0.1cm] (3/4,1/4) to[out=50, in=130] (4.5/4,1/4);
\end{tikzpicture}} 
(\widehat\lambda_{j} \big(\{\lambda_l\}_{l\leq j}) \big) 
 + q^{3}\prod_{j=1}^{5}
 \gamma_{\begin{tikzpicture}
\draw (0,0) -- (1.5/4,0) -- (0.75/4,2/4) -- (0,0) ;
 \fill[fill=black] (0,0) circle (0.1/4);
 \fill[fill=black] (1.5/4,0) circle (0.1/4);
 \fill[fill=black] (0.75/4,2/4) circle (0.1/4);
\fill[fill=black] (3/4,1/4) circle (0.1/4);
\fill[fill=black] (4.5/4,1/4) circle (0.1/4);
\end{tikzpicture}} \;\; 
(\widehat\lambda_{j} \big(\{\lambda_l\}_{l\leq j}) \big) 
\cr\cr\cr
&& = 
\Big[  q^{2} \alpha(\lambda_1)
\alpha(\lambda_3 + \eps \lambda_1)
\alpha(\lambda_5 + \eps \lambda_3 + \eps(1+\eps )\lambda_1)   \Big]
 \cr\cr
&&
\times \Big[\alpha(\lambda_2)
\alpha(\lambda_4 + \eps \lambda_2)
+
\beta(\lambda_2) 
\alpha(\lambda_4 + \eps' \lambda_2)
+
\alpha(\lambda_2)
\beta(\lambda_4 + \eps \lambda_2)
+
q \beta(\lambda_2)
\beta(\lambda_4 + \eps' \lambda_2)\Big] .
\eea
We can already read that the polynomial associated with 
the graph $G_2$ is given by
\bea
P^{\eps,\eps'}(G_2;q;\{\lambda_j\}_{j=2,4})
&=& q\Big[\alpha(\lambda_2)
\alpha(\lambda_4 + \eps \lambda_2)
+
\beta(\lambda_2) 
\alpha(\lambda_4 + \eps' \lambda_2) \crcr
&+&
\alpha(\lambda_2)
\beta(\lambda_4 + \eps \lambda_2)
+
q \beta(\lambda_2)
\beta(\lambda_4 + \eps' \lambda_2)\Big] . 
\eea
Now we can calculate the following terms
\bea
&&
 q^{2}\prod_{j=1}^{5}
 \gamma_{\begin{tikzpicture}
\draw (0,0) -- (1.5/4,0) -- (0.75/4,2/4) ;
 \fill[fill=black] (0,0) circle (0.1/4);
 \fill[fill=black] (1.5/4,0) circle (0.1/4);
 \fill[fill=black] (0.75/4,2/4) circle (0.1/4);
\fill[fill=black] (3/4,1/4) circle (0.1/4);
\fill[fill=black] (4.5/4,1/4) circle (0.1/4);
 \draw[shorten <= 0.1cm, shorten >= 0.1cm] (3/4,1/4) to[out=50, in=130] (4.5/4,1/4);
\draw[shorten <= 0.1cm, shorten >= 0.1cm] (3/4,1/4) to[out=-50, in=-130] (4.5/4,1/4);
\end{tikzpicture}} 
(\widehat\lambda_{j} \big(\{\lambda_l\}_{l\leq j}) \big) 
 + q^{2}\prod_{j=1}^{5}
 \gamma_{\begin{tikzpicture}
\draw (0,0) -- (1.5/4,0) -- (0.75/4,2/4);  
 \fill[fill=black] (0,0) circle (0.1/4);
 \fill[fill=black] (1.5/4,0) circle (0.1/4);
 \fill[fill=black] (0.75/4,2/4) circle (0.1/4);
\fill[fill=black] (3/4,1/4) circle (0.1/4);
\fill[fill=black] (4.5/4,1/4) circle (0.1/4);
\draw[shorten <= 0.1cm, shorten >= 0.1cm] (3/4,1/4) to[out=-50, in=-130] (4.5/4,1/4);
\end{tikzpicture}} \;\; 
(\widehat\lambda_{j} \big(\{\lambda_l\}_{l\leq j}) \big) \cr\cr
&&
+  
 q^{2}\prod_{j=1}^{5}
 \gamma_{\begin{tikzpicture}
\draw (0,0) -- (1.5/4,0) -- (0.75/4,2/4);
 \fill[fill=black] (0,0) circle (0.1/4);
 \fill[fill=black] (1.5/4,0) circle (0.1/4);
 \fill[fill=black] (0.75/4,2/4) circle (0.1/4);
\fill[fill=black] (3/4,1/4) circle (0.1/4);
\fill[fill=black] (4.5/4,1/4) circle (0.1/4);
 \draw[shorten <= 0.1cm, shorten >= 0.1cm] (3/4,1/4) to[out=50, in=130] (4.5/4,1/4);
\end{tikzpicture}} 
(\widehat\lambda_{j} \big(\{\lambda_l\}_{l\leq j}) \big) 
 + q^{3}\prod_{j=1}^{5}
 \gamma_{\begin{tikzpicture}
\draw (0,0) -- (1.5/4,0) -- (0.75/4,2/4); 
 \fill[fill=black] (0,0) circle (0.1/4);
 \fill[fill=black] (1.5/4,0) circle (0.1/4);
 \fill[fill=black] (0.75/4,2/4) circle (0.1/4);
\fill[fill=black] (3/4,1/4) circle (0.1/4);
\fill[fill=black] (4.5/4,1/4) circle (0.1/4);
\end{tikzpicture}} \;\; 
(\widehat\lambda_{j} \big(\{\lambda_l\}_{l\leq j}) \big) 
\cr\cr
&& 
+\Bigg\{ {\mbox{and the like obtained by replacing  \begin{tikzpicture}
\draw (0,0) -- (1.5/4,0) -- (0.75/4,2/4); 
 \fill[fill=black] (0,0) circle (0.1/4);
 \fill[fill=black] (1.5/4,0) circle (0.1/4);
 \fill[fill=black] (0.75/4,2/4) circle (0.1/4);
\end{tikzpicture} by  \begin{tikzpicture}
\draw (1.5/4,0) -- (0.75/4,2/4) -- (0,0) ;
 \fill[fill=black] (0,0) circle (0.1/4);
 \fill[fill=black] (1.5/4,0) circle (0.1/4);
 \fill[fill=black] (0.75/4,2/4) circle (0.1/4);
\end{tikzpicture} 
,  \begin{tikzpicture}
\draw (1.5/4,0) -- (0,0) --  (0.75/4,2/4);
 \fill[fill=black] (0,0) circle (0.1/4);
 \fill[fill=black] (1.5/4,0) circle (0.1/4);
 \fill[fill=black] (0.75/4,2/4) circle (0.1/4);
\end{tikzpicture}  }} \Bigg\}   \cr\cr
&&
= 
P^{\eps,\eps'}(G_2;q;\{\lambda_j\}_{j=2,4}) \Big[  q \beta(\lambda_1)
\alpha(\lambda_3 + \eps' \lambda_1)
\alpha(\lambda_5 + \eps \lambda_3 + \eps'(1+\eps )\lambda_1) \cr\cr
&&
+
 q \alpha(\lambda_1)
\alpha(\lambda_3 + \eps \lambda_1)
\beta(\lambda_5 + \eps \lambda_3 + \eps(1+\eps )\lambda_1)\cr\cr
&&
+
 q \alpha(\lambda_1)
\beta(\lambda_3 + \eps\lambda_1)
\alpha(\lambda_5 + \eps' \lambda_3 + \eps(1+\eps' )\lambda_1)\Big] .
\eea
We deal with the next type of terms of the form:
\bea
&&
 q^{3}\prod_{j=1}^{5}
 \gamma_{\begin{tikzpicture}
\draw (0,0) -- (1.5/4,0);  
 \fill[fill=black] (0,0) circle (0.1/4);
 \fill[fill=black] (1.5/4,0) circle (0.1/4);
 \fill[fill=black] (0.75/4,2/4) circle (0.1/4);
\fill[fill=black] (3/4,1/4) circle (0.1/4);
\fill[fill=black] (4.5/4,1/4) circle (0.1/4);
 \draw[shorten <= 0.1cm, shorten >= 0.1cm] (3/4,1/4) to[out=50, in=130] (4.5/4,1/4);
\draw[shorten <= 0.1cm, shorten >= 0.1cm] (3/4,1/4) to[out=-50, in=-130] (4.5/4,1/4);
\end{tikzpicture}} 
(\widehat\lambda_{j} \big(\{\lambda_l\}_{l\leq j}) \big) 
 + q^{3}\prod_{j=1}^{5}
 \gamma_{\begin{tikzpicture}
\draw (0,0) -- (1.5/4,0);  
 \fill[fill=black] (0,0) circle (0.1/4);
 \fill[fill=black] (1.5/4,0) circle (0.1/4);
 \fill[fill=black] (0.75/4,2/4) circle (0.1/4);
\fill[fill=black] (3/4,1/4) circle (0.1/4);
\fill[fill=black] (4.5/4,1/4) circle (0.1/4);
\draw[shorten <= 0.1cm, shorten >= 0.1cm] (3/4,1/4) to[out=-50, in=-130] (4.5/4,1/4);
\end{tikzpicture}} \;\; 
(\widehat\lambda_{j} \big(\{\lambda_l\}_{l\leq j}) \big) \cr\cr
&&
+ 
 q^{3}\prod_{j=1}^{5}
 \gamma_{\begin{tikzpicture}
\draw (0,0) -- (1.5/4,0);   
 \fill[fill=black] (0,0) circle (0.1/4);
 \fill[fill=black] (1.5/4,0) circle (0.1/4);
 \fill[fill=black] (0.75/4,2/4) circle (0.1/4);
\fill[fill=black] (3/4,1/4) circle (0.1/4);
\fill[fill=black] (4.5/4,1/4) circle (0.1/4);
 \draw[shorten <= 0.1cm, shorten >= 0.1cm] (3/4,1/4) to[out=50, in=130] (4.5/4,1/4);
\end{tikzpicture}} 
(\widehat\lambda_{j} \big(\{\lambda_l\}_{l\leq j}) \big) 
 + q^{4}\prod_{j=1}^{5}
 \gamma_{\begin{tikzpicture}
\draw (0,0) -- (1.5/4,0);  
 \fill[fill=black] (0,0) circle (0.1/4);
 \fill[fill=black] (1.5/4,0) circle (0.1/4);
 \fill[fill=black] (0.75/4,2/4) circle (0.1/4);
\fill[fill=black] (3/4,1/4) circle (0.1/4);
\fill[fill=black] (4.5/4,1/4) circle (0.1/4);
\end{tikzpicture}} \;\; 
(\widehat\lambda_{j} \big(\{\lambda_l\}_{l\leq j}) \big) 
\cr\cr
&& 
+ \Big\{ {\mbox{and the like obtained by replacing  \begin{tikzpicture}
\draw (0,0) -- (1.5/4,0);  
 \fill[fill=black] (0,0) circle (0.1/4);
 \fill[fill=black] (1.5/4,0) circle (0.1/4);
 \fill[fill=black] (0.75/4,2/4) circle (0.1/4);
\end{tikzpicture} by  \begin{tikzpicture}
\draw (0.75/4,2/4) -- (0,0) ;
 \fill[fill=black] (0,0) circle (0.1/4);
 \fill[fill=black] (1.5/4,0) circle (0.1/4);
 \fill[fill=black] (0.75/4,2/4) circle (0.1/4);
\end{tikzpicture} 
,  \begin{tikzpicture}
\draw (1.5/4,0) --  (0.75/4,2/4);
 \fill[fill=black] (0,0) circle (0.1/4);
 \fill[fill=black] (1.5/4,0) circle (0.1/4);
 \fill[fill=black] (0.75/4,2/4) circle (0.1/4);
\end{tikzpicture}  }} \Big\} \cr\cr
&& = 
P^{\eps,\eps'}(G_2;q;\{\lambda_j\}_{j=2,4}) \Big[  q^2 \beta(\lambda_1)
\beta(\lambda_3 + \eps' \lambda_1)
\alpha(\lambda_5 + \eps' \lambda_3 + \eps'(1+\eps' )\lambda_1) \cr\cr
&&
+
 q^2 \alpha(\lambda_1)
\beta(\lambda_3 + \eps \lambda_1)
\beta(\lambda_5 + \eps' \lambda_3 + \eps(1+\eps' )\lambda_1)\cr\cr
&&
+
 q^2 \beta(\lambda_1)
\alpha(\lambda_3 + \eps'\lambda_1)
\beta(\lambda_5 + \eps\lambda_3 + \eps'(1+\eps )\lambda_1)\Big] .
\eea
Finally, we evaluate 
\bea
&&
 q^{4}\prod_{j=1}^{5}
 \gamma_{\begin{tikzpicture}
 \fill[fill=black] (0,0) circle (0.1/4);
 \fill[fill=black] (1.5/4,0) circle (0.1/4);
 \fill[fill=black] (0.75/4,2/4) circle (0.1/4);
\fill[fill=black] (3/4,1/4) circle (0.1/4);
\fill[fill=black] (4.5/4,1/4) circle (0.1/4);
 \draw[shorten <= 0.1cm, shorten >= 0.1cm] (3/4,1/4) to[out=50, in=130] (4.5/4,1/4);
\draw[shorten <= 0.1cm, shorten >= 0.1cm] (3/4,1/4) to[out=-50, in=-130] (4.5/4,1/4);
\end{tikzpicture}} 
(\widehat\lambda_{j} \big(\{\lambda_l\}_{l\leq j}) \big) 
 + q^{4}\prod_{j=1}^{5}
 \gamma_{\begin{tikzpicture}
 \fill[fill=black] (0,0) circle (0.1/4);
 \fill[fill=black] (1.5/4,0) circle (0.1/4);
 \fill[fill=black] (0.75/4,2/4) circle (0.1/4);
\fill[fill=black] (3/4,1/4) circle (0.1/4);
\fill[fill=black] (4.5/4,1/4) circle (0.1/4);
\draw[shorten <= 0.1cm, shorten >= 0.1cm] (3/4,1/4) to[out=-50, in=-130] (4.5/4,1/4);
\end{tikzpicture}} \;\; 
(\widehat\lambda_{j} \big(\{\lambda_l\}_{l\leq j}) \big) \cr\cr
&&
+  
 q^{4}\prod_{j=1}^{5}
 \gamma_{\begin{tikzpicture}
 \fill[fill=black] (0,0) circle (0.1/4);
 \fill[fill=black] (1.5/4,0) circle (0.1/4);
 \fill[fill=black] (0.75/4,2/4) circle (0.1/4);
\fill[fill=black] (3/4,1/4) circle (0.1/4);
\fill[fill=black] (4.5/4,1/4) circle (0.1/4);
 \draw[shorten <= 0.1cm, shorten >= 0.1cm] (3/4,1/4) to[out=50, in=130] (4.5/4,1/4);
\end{tikzpicture}} 
(\widehat\lambda_{j} \big(\{\lambda_l\}_{l\leq j}) \big) 
 + q^{5}\prod_{j=1}^{5}
 \gamma_{\begin{tikzpicture}
 \fill[fill=black] (0,0) circle (0.1/4);
 \fill[fill=black] (1.5/4,0) circle (0.1/4);
 \fill[fill=black] (0.75/4,2/4) circle (0.1/4);
\fill[fill=black] (3/4,1/4) circle (0.1/4);
\fill[fill=black] (4.5/4,1/4) circle (0.1/4);
\end{tikzpicture}} \;\; 
(\widehat\lambda_{j} \big(\{\lambda_l\}_{l\leq j}) \big) 
\cr\cr
&& 
P^{\eps,\eps'}(G_2;q;\{\lambda_j\}_{j=2,4}) \Big[  q^3 \beta(\lambda_1)
\beta(\lambda_3 + \eps' \lambda_1)
\beta(\lambda_5 + \eps' \lambda_3 + \eps'(1+\eps' )\lambda_1)\Big].
\eea
Adding all contributions, one finds the ordering-dependent
Tutte polynomial associated with the graph $G_1 \cup G_2$,
performing a sequence of contractions and deletions as $(1,2,3,4,5)$. 

In a final comment, we want to address here the special case $\eps=\eps'=1$
and $\lambda_j = \lambda$ in view of possible interesting deformation of 
the ordinary Tutte polynomial. Consider the above example on $G_2$ and $G_1$, setting for simplicity 
$\beta=1$, and therefore we write 
\bea
P^{1,1}(G_2;q;\lambda)
&=& q\Big[\alpha(\lambda)
\alpha(2\lambda)
+
\alpha(2\lambda) +
\alpha(\lambda)
+
q \Big]  \crcr
&=& \sum_{B \subset G_2} q^{k(B)}
\prod_{e\in B}\alpha(\lambda^B_e)\,, 
\eea  
where $\lambda_e^B = c^B_e \lambda$ and  $c_e^B$ is  
a positive integer. Observe that, for this example,  $c_e^B$ becomes independent of $B$ and so we reduce to the multivariate version of the Tutte polynomial \cite{Sokal}  (to be explicit, we can redefine the edge labelling such that the polynomial $P^{1,1}$ coincides with multivariate Tutte polynomial).
 This is generally true  
when the adjacency matrix of the line graph is invariant under permutation of the edges. The case of clique graphs is a particular
instance for which this is realized.
Thus, for this type of graphs, we naturally loose the dependence on the order in which we perform the sequence of contraction deletion. 
We can check if this is also the case for the graph $G_1$ 
and find,
\bea
P^{1,1}(G_1;q;\lambda)
 &=&   q^{2} \alpha(\lambda)
\alpha(2\lambda)
\alpha(4\lambda)   
+ 
 q 
\alpha(2\lambda)
\alpha(4\lambda) 
+
 q \alpha(\lambda)
\alpha(2\lambda)
\cr\cr
&+&
 q \alpha(\lambda)
\alpha(4\lambda) 
+
  q^2
\alpha(4\lambda) 
+
 q^2 \alpha(\lambda)
+
 q^2 
\alpha(2\lambda)
 +
 q^3 \cr\cr
&=& \sum_{B \subset G_1} q^{k(B)}
\prod_{e\in B}\alpha(\lambda_e),
\eea
confirming our expectations. In conclusion, in this particular limit, and
for particular graphs, the order-independence of the generalized Tutte polynomial 
can be effectively restored.

\end{document}